\documentclass[
amsmath, amssymb, 
aps, pra,
reprint, twocolumn,
superscriptaddress,
longbibliography,
floatfix,
nofootinbib
]{revtex4}

\usepackage[a4paper, left=0.8in, right=0.8in, top=0.9in, bottom=0.9in]{geometry}

\usepackage[utf8]{inputenc}
\usepackage{amsmath,amssymb,amsthm}
\usepackage{braket}
\usepackage[urlcolor=blue, colorlinks=true,citecolor=blue,linkcolor=black]{hyperref}
\usepackage{cleveref}
\usepackage{graphicx}
\usepackage[caption=false]{subfig}
\usepackage{float}
\usepackage{tabularx}
\usepackage{amsfonts}
\usepackage{xcolor}
\usepackage{algorithm}
\usepackage{algpseudocode}

\newtheorem{theorem}{Theorem}

\newtheorem{definition}{Definition}
\newtheorem{corollary}{Corollary}

\begin{document}

\title{Equivariant quantum circuits for learning on weighted graphs}

\author{Andrea Skolik$^*$}
    \affiliation{Leiden University, Niels Bohrweg 1, 2333 CA Leiden, The Netherlands}
    \affiliation{Volkswagen Data:Lab, Ungererstra\ss e 69, 80805 Munich, Germany}
    \email[]{Corresponding author a.skolik@liacs.leidenuniv.nl}
    
\author{Michele Cattelan}
    \affiliation{Volkswagen Data:Lab, Ungererstra\ss e 69, 80805 Munich, Germany}
    \affiliation{Institute for Theoretical Physics, University of Innsbruck, A-6020 Innsbruck, Austria}
    
    \author{Sheir Yarkoni}
    \affiliation{Leiden University, Niels Bohrweg 1, 2333 CA Leiden, The Netherlands}
    \affiliation{Volkswagen Data:Lab, Ungererstra\ss e 69, 80805 Munich, Germany}
    
\author{Thomas Bäck}
    \affiliation{Leiden University, Niels Bohrweg 1, 2333 CA Leiden, The Netherlands}
    
\author{Vedran Dunjko}
\affiliation{Leiden University, Niels Bohrweg 1, 2333 CA Leiden, The Netherlands}
    
\date{\today}

\begin{abstract}
Variational quantum algorithms are the leading candidate for advantage on near-term quantum hardware. When training a parametrized quantum circuit in this setting to solve a specific problem, the choice of ansatz is one of the most important factors that determines the trainability and performance of the algorithm. In quantum machine learning (QML), however, the literature on ansatzes that are motivated by the training data structure is scarce. In this work, we introduce an ansatz for learning tasks on weighted graphs that respects an important graph symmetry, namely equivariance under node permutations. We evaluate the performance of this ansatz on a complex learning task, namely neural combinatorial optimization, where a machine learning model is used to learn a heuristic for a combinatorial optimization problem. We analytically and numerically study the performance of our model, and our results strengthen the notion that symmetry-preserving ansatzes are a key to success in QML.
\end{abstract}

\maketitle

\section{Introduction}
Hybrid quantum-classical algorithms in which a parametrized quantum circuit (PQC) is optimized by a classical algorithm to solve a specific problem, also known as variational quantum algorithms \cite{cerezo2021variational}, are expected to be the leading candidate for near-term quantum advantage due to their flexibility and the hope that their hybrid nature can make them robust to noise to some degree.  These types of algorithms can be applied in a variety of contexts, and it is known that the right choice of circuit structure, also known as the ansatz, is of key importance for the performance of these models. Much work has been dedicated to understand how circuits have to be structured to address problems in optimization \cite{farhi2014quantum,bennett2021quantum} or chemistry \cite{grimsley2018adapt,peruzzo2014variational}. For quantum machine learning however, it is largely unknown which type of ansatz should be used for a given type of data. In absence of an informed choice, general architectures as the hardware-efficient ansatz \cite{kandala2017hardware} are often used \cite{benedetti2019parameterized}. It is known that ansatzes with randomly selected structure scale badly as the width and depth of the circuit grows, most prominently because of the barren plateau phenomenon \cite{mcclean2018barren,arrasmith2021effect,marrero2021entanglement} where the gradients of a PQC vanish exponentially as the system size grows and thus render training impossible. This situation can be compared to the early days of neural networks (NNs), where fully connected feedforward NNs were the standard architecture. These types of NNs suffer from similar trainability issues as random quantum circuits \cite{smagt1998feed}. Recent breakthroughs in deep learning were in part possible because more efficient architectures that are directly motivated by the training data structure have been developed \cite{silver2016mastering,jumper2021highly,ramesh2021zero}. 
\begin{figure}[t]
\includegraphics[width=0.5\textwidth]{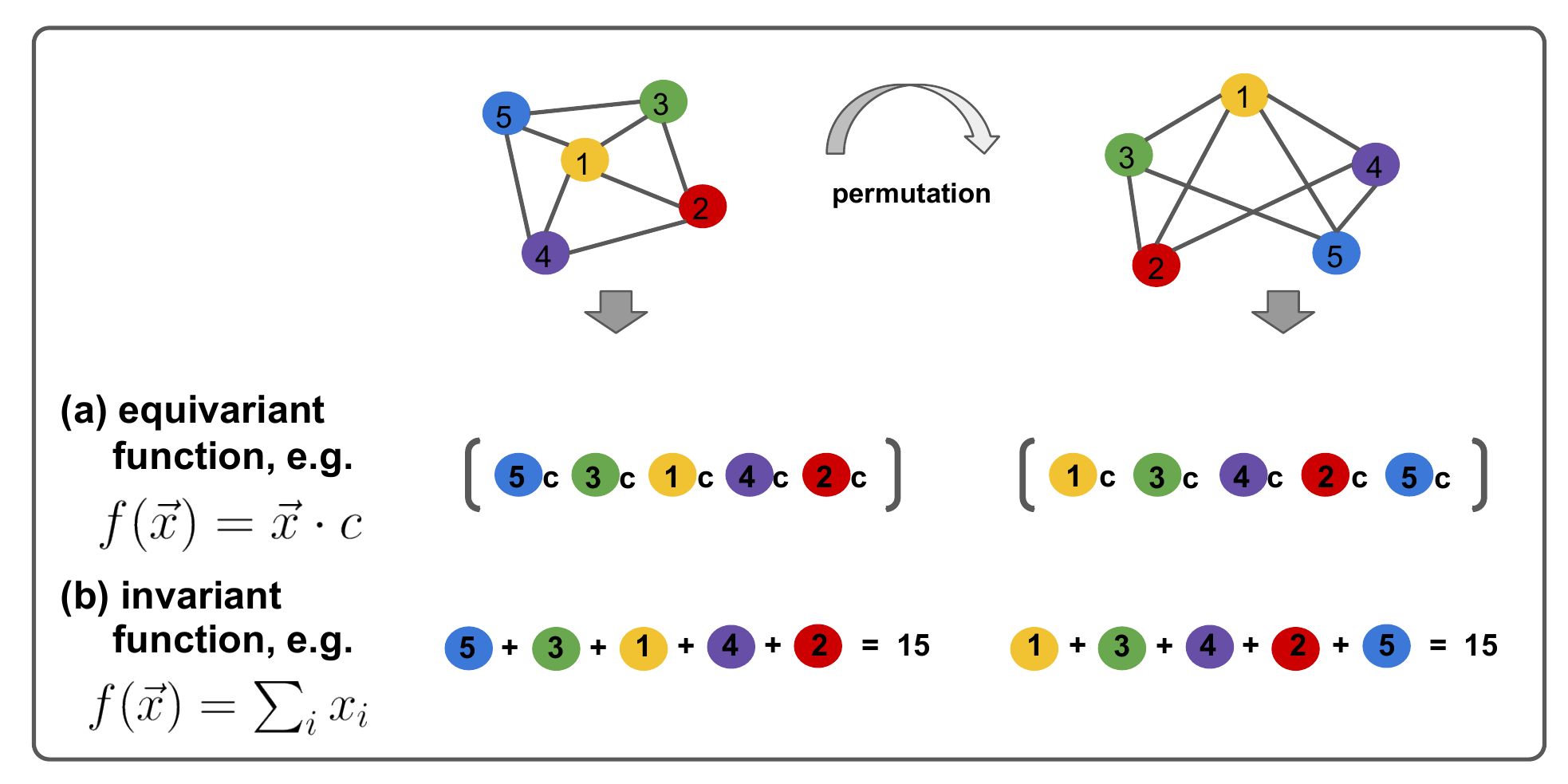}
\caption{Depiction of two functions that respect important symmetries of graphs. a) The permutation equivariant function will yield the same output values for each graph permutation, but reordered according to the reordering of nodes. The above example shows a simple function that takes node features as an input and multiplies them with a constant. b) An invariant function will yield the same output, regardless of the permutation. The above example shows a simple function that takes node features as input and computes their sum. Which type of symmetry is preferable depends on the task at hand.}
\label{fig:symmetries}
\end{figure}
In fact, a whole field that studies the mathematical properties of successful NN architectures has emerged in the past decade, known as geometric deep learning. This field studies the properties of common NN architectures, like convolutional NNs or graph NNs, through the lens of group theory and geometry and provides an understanding of why these structured types of models are the main drivers of recent advances in deep learning. The success of these models can largely be attributed to the fact that they preserve certain symmetries that are present in the training data. Graph NNs, for example, take graph-structured data as input and their layers are designed such that they respect one of two important graph symmetries: invariance or equivariance under permutation of vertices \cite{zhou2020graph}, as depicted in \Cref{fig:symmetries}. Graph-structured data is ubiquitous in real-world problems, for example to predict properties of molecules \cite{jumper2021highly} or to solve combinatorial optimization problems \cite{cappart2021combinatorial}. Even images can be viewed as special types of graphs, namely those defined on a lattice with nearest-neighbor connections. This makes graph NNs applicable in a multitude of contexts, and motivated a number of works that study quantum versions of these models \cite{verdon2019quantum,henry2021quantum,zheng2021quantum,mernyei2021equivariant}. However, the key questions of how to design symmetry-preserving ansatzes motivated by a concrete input data structure and how these ansatzes perform compared to those that are structurally unrelated to the given learning problem remain open.

In this work, we address these open questions by introducing a symmetry-preserving ansatz for learning problems where the training data is given in form of weighted graphs, and study its performance both numerically and analytically. To do this, we extend the family of ansatzes from \cite{mernyei2021equivariant} to incorporate weighted edges of the input graphs and prove that the resulting ansatz is equivariant under node permutations. To evaluate this ansatz on a complex learning task where preserving a given symmetry can yield a significant performance advantage, we apply it in a domain where classical graph NNs have been used extensively: neural combinatorial optimization \cite{cappart2021combinatorial}. In this setting, a model is trained to solve instances of a combinatorial optimization problem. Namely, we train our proposed ansatz to find approximate solutions to the Traveling Salesperson Problem (TSP). We numerically compare our ansatz to three non-equivariant ansatzes on instances with up to 20 cities (20 qubits), and show that the more the equivariance property of the ansatz is broken, the worse performance becomes and that a simple hardware-efficient ansatz completely fails on this learning task. Additionally, we analytically study the expressivity of our model at depth one, and show under which conditions there exists a parameter setting for any given TSP instance of arbitrary size for our ansatz that produces the optimal tour with the learning scheme that is applied in this work.

The neural combinatorial optimization approach presented in this work also provides an alternative method to employ near-term quantum computers to tackle combinatorial optimization problems. As problem instances are directly encoded into the circuit in form of graphs without the need to specify a cost Hamiltonian, this approach is even more frugal than that of the quantum approximate optimization algorithm (QAOA) \cite{farhi2014quantum} in terms of the requirement on the number of qubits and connectivity in cases where the problem encoding is non-trivial. For the TSP specifically, standard Hamiltonian encodings require $n^2$ variables where $n$ is the number of cities (or $n \log(n)$ variables at the cost of increased circuit depth) \cite{lucas2014ising}, whereas our approach requires only $n$ qubits and two-body interactions. We do note that the theoretical underpinnings and expected guarantees of performance of our method are very different and less rigorous than those of the QAOA, so the two are hard to compare directly. However, we establish a theoretical connection to the QAOA based on the structure of our ansatz, and in addition numerically compare QAOA performance on TSP instances with 5 cities to the performance of the proposed neural combinatorial optimization approach. We find that our ansatz at depth one outperforms the QAOA even at depth up to three. From a pragmatic point of view, linear scaling in qubit numbers w.r.t.\ number of problem variables, as opposed to e.g.\ quadratic scaling as in the case of the TSP, dramatically changes the applicability of quantum algorithms in the near- to mid-term.

Our work illustrates the merit of using symmetry-preserving ansatzes for QML on the example of graph-based learning, and underlines the notion that in order to successfully apply variational quantum algorithms for ML tasks in the future, the usage of ansatzes unrelated to the problem structure, which are popular in current QML research, is limited as problem sizes grow. This work motivates further study of ``geometric quantum learning" in the vein of the classical field of geometric deep learning, to establish more effective ansatzes for QML, as these are a prerequisite to efficiently apply quantum models on any practically relevant learning task in the near-term.

\section{Results}

In this section, we formally introduce the structure of our equivariant quantum circuit (EQC) for learning tasks on weighted graphs that we use in this work. Examples of graph-structured data that can be used as input in this type of learning task are images \cite{ramesh2021zero}, social networks \cite{fan2019graph} or molecules \cite{jumper2021highly}. In general, when learning based on graph data, there are two sets of features: node features and edge features. Depending on the specific learning task, it might be enough to use only one set of these features as input data, and the specific implementation of the circuit will change accordingly. As mentioned above, an example of an ansatz for cases where encoding node features suffices is the family of ansatzes introduced in \cite{mernyei2021equivariant}. In our case, we use both node and edge features to solve TSP instances. In case of the nodes, we encode whether a node (city) is already present in the partial tour at time step $t$ to inform the node selection process described later in \Cref{def:node_selection}. For the edges, we simply encode the edge weights of the graph as these correspond to the distances between nodes in the TSP instance's graph. In this work, we use one qubit per node in the graph, but in general multiple qubits per node are also possible. We discuss the details of this in the supplementary material. We now proceed to define the ansatz in terms of encoding node information in form of $\boldsymbol \alpha$ (see \Cref{def:annotated_graph}) and edge information in terms of the weighted graph edges $\varepsilon_{ij} \in \mathcal{E}$. For didactic reasons we relate the node and edge features to the concrete learning task that we seek to solve in this work, however, we note that this encoding scheme is applicable in the context of other learning tasks on weighted graphs as well.

\subsection{Equivariant quantum circuit}
\label{sec:ansatz_structure_equivariance}

Given a graph $\mathcal{G}(\mathcal{V}, \mathcal{E})$ with node features $\boldsymbol \alpha$ and weighted edges $ \mathcal{E}$, and trainable parameters $\boldsymbol \beta, \boldsymbol \gamma \in \mathbb{R}^{p} $, our ansatz at depth $p$ is of the following form
\begin{multline}\label{eq:ansatz}
    \ket{\mathcal{E}, \boldsymbol \alpha, \boldsymbol \beta, \boldsymbol \gamma}_p = U_N(\boldsymbol \alpha, \beta_p) U_G(\mathcal{E}, \gamma_p) \\ \dots U_N(\boldsymbol \alpha, \beta_1) U_G(\mathcal{E}, \gamma_1) \ket{s},
\end{multline}
where $\ket{s}$ is the uniform superposition of bitstrings of length $n$,
\begin{equation}
    \ket{s} = \frac{1}{\sqrt{2^n}} \sum_{x \in \{0, 1\}^n} \ket{x},
\end{equation}
$U_N(\boldsymbol \alpha, \beta_j)$ with $\mathrm{Rx}(\theta) = e^{-i \frac{\theta}{2} X}$, is defined as
\begin{equation}\label{eq:mixer}
    U_N(\boldsymbol \alpha, \beta_j) = \bigotimes_{l=1}^n \mathrm{Rx}(\alpha_l \cdot \beta_j),
\end{equation}
and $U_G(\mathcal{E}, \gamma_j)$ is
\begin{equation}
    U_G(\mathcal{E}, \gamma_j) = \mathrm{exp}(-i \gamma_j H_\mathcal{G})
\end{equation}
with $H_\mathcal{G} = \sum_{(i,j) \in \mathcal{E}} \varepsilon_{ij} \sigma_z^{(i)} \sigma_z^{(j)}$ and $\mathcal{E}$ are the edges of graph $\mathcal{G}$ weighted by $\varepsilon_{ij}$. A 5-qubit example of this ansatz can be seen in \Cref{fig:eqc}.

For $p=1$, we have
\begin{multline}\label{eq:extended_ansatz}
    \ket{\mathcal{E}, \boldsymbol \alpha, \beta, \gamma}_1 = U_N(\boldsymbol \alpha, \beta) U_G(\mathcal{E}, \gamma) \ket{s} \\
    = \frac{1}{\sqrt{2^n}} \\
    \cdot \sum_{x \in \{0, 1\}^n} 
    \underbrace{\left(\mathrm{cos} \frac{\alpha_1 \beta }{2} + \dots -i \mathrm{sin} \frac{\alpha_l \beta }{2} - \dots -i\mathrm{sin}\frac{\alpha_n \beta }{2} \right)}_\text{weighted bitflip terms}\\
    \cdot \underbrace{\mathrm{exp}\left(\sum_{(i,j) \in \mathcal{E}} \mathrm{diag}(Z_i Z_j)_{\ket{x}} \cdot -i \ \frac{\pi}{2} \gamma \varepsilon_{ij}\right)}_\text{edge weights} \ket{x},
\end{multline}
where $\mathrm{diag}(Z_i Z_j)_{\ket{x}} = \pm 1$ is the entry in the matrix corresponding to each $Z_i Z_j$ term, e.g., $I_1 \otimes \dots \otimes Z_i \otimes I_k \otimes \dots \otimes Z_j \otimes \dots \otimes I_n$, corresponding to the basis state $\ket{x}$. (E.g., the first term on the diagonal corresponds to the all-zero state, and so on.) We see that the first group of terms, denoted weighted bitflip terms, is a sum over products of terms that encode the node features. In other words, in the one-qubit case we get a sum over sine and cosine terms, in the two-qubit case we get a sum over products of pairs of sine and cosine terms, and so on. The terms in the second part of the equation denoted edge weights is the exponential of a sum over edge weight terms. As we start in the uniform superposition, each basis state's amplitude depends on all node and edge features, but with different signs and therefore different terms interfering constructively and destructively for every basis state. This can be regarded as a quantum version of the aggregation functions used in classical graph NNs, where the $k$-th layer of a NN aggregates information over the $k$-local neighborhood of the graph in a permutation equivariant way \cite{bronstein2021geometric}. In a similar fashion, the terms in \Cref{eq:extended_ansatz} aggregate node and edge information and become more complex with each additional layer in the PQC.

\begin{figure}[t]
\hspace*{-0.2in}
\vspace*{-0.3in}
\includegraphics[width=0.5\textwidth]{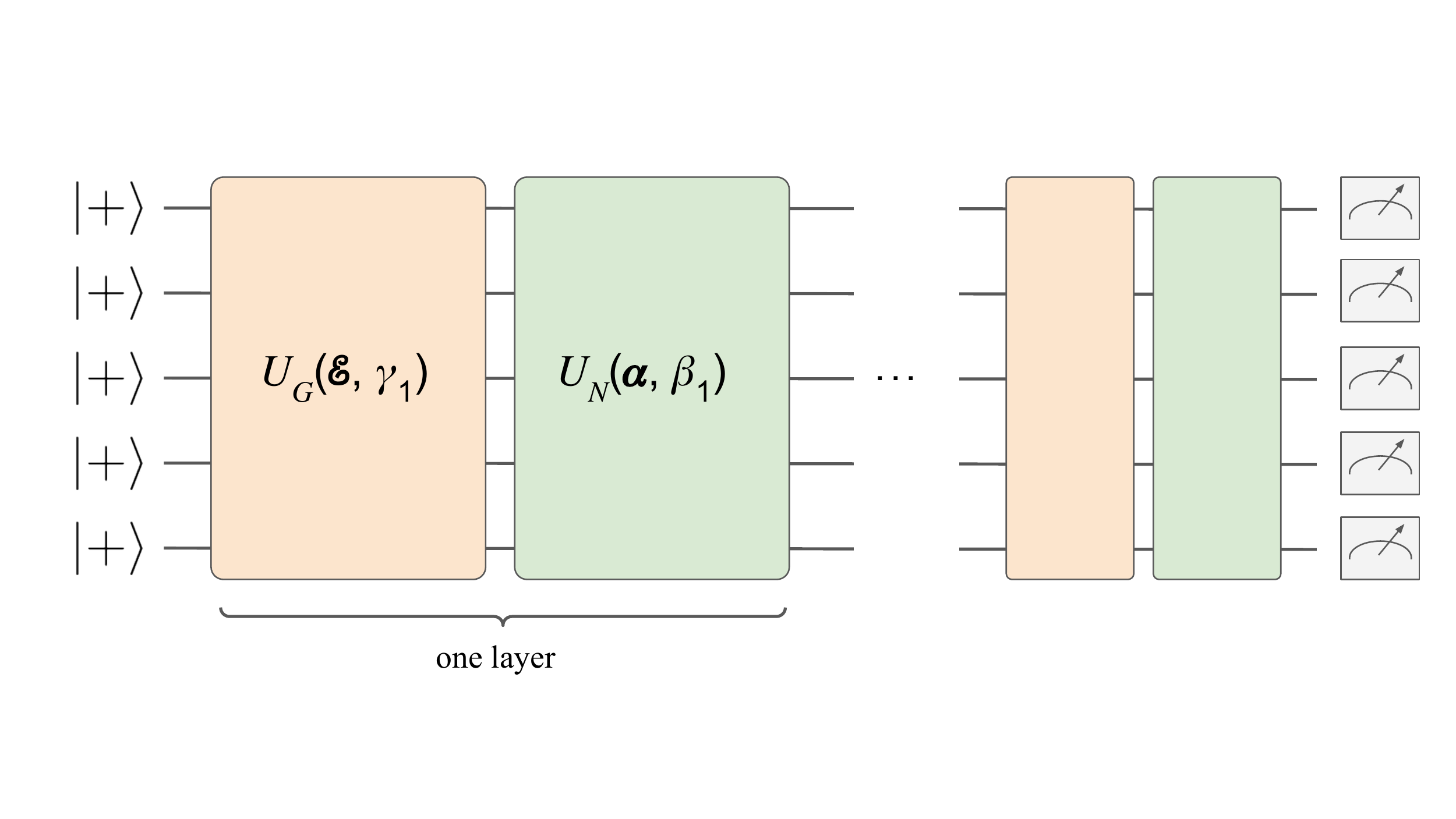}
\caption{Ansatz used in this work. Each layer consists of two parts: the first part $U_G$ encodes edge features, while the second part $U_N$ encodes node features. Each of the two parts is parametrized by one parameter $\beta_l, \gamma_l$, respectively.}
\label{fig:eqc}
\end{figure}

The reader may already have observed that this ansatz is closely related to an ansatz that is well-known in quantum optimization: that of the quantum approximate optimization algorithm \cite{farhi2014quantum}. Indeed, our ansatz can be seen as a special case of the QAOA, where instead of using a cost Hamiltonian to encode the problem, we directly encode instances of graphs and apply the ``mixer terms" in \Cref{eq:mixer} only to nodes not yet in the partial tour. This correspondence will later let us use known results for QAOA-type ansatzes at depth one \cite{ozaeta2022expectation} to derive exact analytical forms of the expectation values of our ansatz, and use these to study its expressivity. 

As our focus is on implementing an ansatz that respects a symmetry that is useful in graph learning tasks, namely an equivariance under permutation of vertices of the input graph, we now show that each part of our ansatz respects this symmetry.

\begin{theorem}[Permutation equivariance of the ansatz]\label{thm:ansatz_perm_equivariance}
Let the ansatz of depth $p$ be of the type as defined in \Cref{eq:ansatz} with initial state $\ket{+}^{\otimes n}$ and parameters $\boldsymbol \beta, \boldsymbol \gamma \in \mathbb{R}^p$, that represents an instance of a graph $\mathcal{G}$ with nodes $\mathcal{V}$ and the list of edges $\mathcal{E}$ with corresponding edge weights $\varepsilon_{ij}$, and node features $\boldsymbol \alpha \in \mathbb{R}^n$ with $n=|\mathcal{V}|$. Let $\sigma$ be a permutation of the vertices in $\mathcal{V}$, $P_{\sigma} \in \mathbb{B}^{n \times n}$ the corresponding permutation matrix that acts on the weighted adjacency matrix $A$ of $\mathcal{G}$, and $\Tilde{P}_{\sigma} \in \mathbb{B}^{2^n \times 2^n}$ a matrix that maps the tensor product $\ket{v_1} \otimes \ket{v_2} \otimes \dots \otimes \ket{v_n}$ with $\ket{v_i} \in \mathbb{C}^2$ to $\ket{v_{\Tilde{p}_{\sigma}(1)}} \otimes \ket{v_{\Tilde{p}_{\sigma}(2)}} \otimes \dots \otimes \ket{v_{\Tilde{p}_{\sigma}(n)}}$. Then, the following relation holds,
\begin{equation}
    \ket{\mathcal{E}_A, \boldsymbol \alpha, \boldsymbol \beta, \boldsymbol \gamma}_p = \Tilde{P}_{\sigma} \ket{\mathcal{E}_{(P_{\sigma}^T A P_{\sigma})},  P_{\sigma}^T \boldsymbol\alpha, \boldsymbol \beta, \boldsymbol \gamma}_p,
\end{equation}
where $\mathcal{E}_{(\cdot)}$ denotes a specific permutation of the adjacency matrix $A$ of the given graph. We call an ansatz that satisfies this property permutation equivariant.
\end{theorem}

As mentioned before, our ansatz is closely related to those in \cite{mernyei2021equivariant}, and the authors of this work prove permutation equivariance of unitaries that are defined in terms of unweighted adjacency matrices of graphs. In order to prove equivariance of our circuit, we have to generalize their result to the case where a weighted graph is encoded in the form of a Hamiltonian, and parametrized by a set of free parameters as described in \Cref{eq:ansatz}. In the non-parametrized case this is trivial, as edge weights and node features are directly permuted as a consequence of the permutation of the graph. When introducing parameters to the node and edge features, however, we have to make sure that the parameters themselves preserve equivariance, as the parameters are not tied to the adjacency matrix but to the circuit itself. To guarantee this, we make the parametrization itself permutation invariant by assigning one node and edge parameter per layer, respectively, and this makes us arrive at the QAOA-type parametrization shown in \Cref{eq:ansatz}. Another difference of our proof to that in \cite{mernyei2021equivariant} is that we consider a complete circuit including its initial state, instead of only guaranteeing that the unitaries that act on the initial state are permutation equivariant. We provide the detailed proof of equivariance of our ansatz in the supplementary material.

The above definition and proof are given in terms of a learning problem where we map one vertex to one qubit directly. However, settings where we require more than one qubit to encode node information are easily possible with this type of architecture as well. In order to preserve equivariance of our ansatz construction, three conditions have to hold: i) the initial state of the circuit has to be permutation invariant or equivariant, ii) the two-qubit gates used to encode edge weights have to commute, iii) the parametrization of the gates has to be permutation invariant. In the case where each vertex or edge is represented with more than just one gate per layer, one has freedom on how to do this as long as the above i)-iii) still hold. A simple example is when each vertex is represented by $m$ qubits: i) the initial state remains to be the uniform superposition, ii) the topology of the two-qubit gates that represent edges has to be changed according to the addition of the new qubits, but $ZZ$-gates can still be used to encode the information, iii) the parametrization is the same as in the one-qubit-per-vertex case.

\subsection{Trainability of ansatz}

Our goal in this work is to introduce a problem-tailored ansatz for a specific data type that provides trainability advantages compared to unstructured ansatzes. One important question that arises in this context is that of barren plateaus, where the variance of derivatives for random circuits vanishes exponentially with the system size \cite{mcclean2018barren}. This effect poses challenges for scaling up circuit architectures like the hardware-efficient ansatz \cite{kandala2017hardware}, as even at a modest number of qubits and layers a quantum model like this can become untrainable \cite{cerezo2021cost,marrero2021entanglement,arrasmith2021effect}. Therefore it is important to address the presence of barren plateaus when introducing an ansatz. In a recent work \cite{schatzki2022theoretical}, it has been proven that barren plateaus are not present in circuits that are equivariant under the symmetric group $S_n$, namely the group of permutations on $n$ elements, in this case all permutations over the qubits. While our circuit is also permutation equivariant, we define permutations based on the input graphs and not the qubits themselves, so our approach differs from the equivariant quantum neural networks in \cite{schatzki2022theoretical} as a) the incorporation of edge weights into the unitaries prevents the unitaries from commuting with all possible permutations of qubits, and b) multiple qubits can potentially correspond to one vertex. While permutation equivariance poses some restrictions on the expressibility of the ansatz and one would expect a better scaling of gradients than in, e.g., hardware-efficient types of circuits, the results of \cite{schatzki2022theoretical} do not directly translate to our work for the above reasons. 

To get additional insight, one can also turn to results on barren plateaus related to QAOA-type circuits, due to the structural similarity that our ansatz has to them. The authors of \cite{larocca2022diagnosing} investigate the scaling of the variance of gradients of two related types of ansatzes. They characterize ansatzes given by the following two Hamiltonians: the transverse field Ising model (TFIM),
\begin{equation}
    H_{\mathrm{TFIM}} = \sum_{i=1}^{n_f} Z_i Z_{i+1} + h_x \sum_{i=1}^n X_i,
\end{equation}
where $n_f = n-1$ ($n_f = n$) for open (periodic) boundary conditions, and a spin glass (SG), 
\begin{equation}
    H_{\mathrm{SG}} = \sum_{i < j} h_i Z_i + J_{ij} Z_i Z_j + \sum_{i=1}^n X_i,
\end{equation}
with $h_i, J_{ij}$ drawn from a Gaussian distribution. Based on the generators of those two ansatzes, the authors of \cite{larocca2022diagnosing} show that an ansatz that consists of layers given by the TFIM Hamiltonian has a favorable scaling of gradients. An ansatz that consists of layers given by $H_{\mathrm{SG}}$, on the other hand, does not. Considering the results for the two above Hamiltonians, one can expect that whether our ansatz exhibits barren plateaus will strongly depend on the encoded graphs, i.e., the connectivity, edge weights and node features. Which types of graphs lead to a favorable scaling of gradients, and for what learning tasks our ansatz exhibits good performance at a number of layers polynomial in the input size, is an interesting question that we leave for future work.

Additionally to barren plateaus that are a result of the randomness of the circuit, there is a type of barren plateau that is caused by hardware noise, called noise-induced barren plateaus (NIBPs) \cite{wang2021noise}. This problem can not be directly mitigated by the choice of circuit architecture, as eventually all circuit architectures are affected by hardware noise, especially when they become deeper. We do not expect that our circuit is resilient to NIBPs, however, the numerical results in \Cref{sec:num_results} show that the EQC already performs well with only one layer for the environment we study in this work as we scale up the problem size. This provides hope that, at least in terms of circuit depth, the EQC will scale favorably in the number of layers as the number of qubits in the circuit is increased, and therefore the effect of NIBPs will be less severe than for other circuit architectures with the same number of qubits.

Another important question for the training of ML models is that of data efficiency, i.e., how many training data points are required to achieve a low generalization error. Indeed, one of the key motivating factors behind the design of geometric models that preserve symmetries in the training data is to reduce the size of the training data set. In the classical literature, it was shown that geometric models require fewer training data and as a result often fewer parameters as models that do not preserve said symmetries \cite{mei2021learning}. Recent work showed that this is also true for $S_n$-equivariant quantum models \cite{schatzki2022theoretical}, where the authors give an improved bound on the generalization error compared to the bounds that were previously shown to exist for general classes of PQCs \cite{caro2022generalization}. However, the results from \cite{schatzki2022theoretical} do again not directly translate to our approach as stated in the context of barren plateaus above.

\subsection{Quantum neural combinatorial optimization with the EQC}
\label{sec:qgnn_architecture}

\begin{figure*}[t]
\includegraphics[width=1\textwidth]{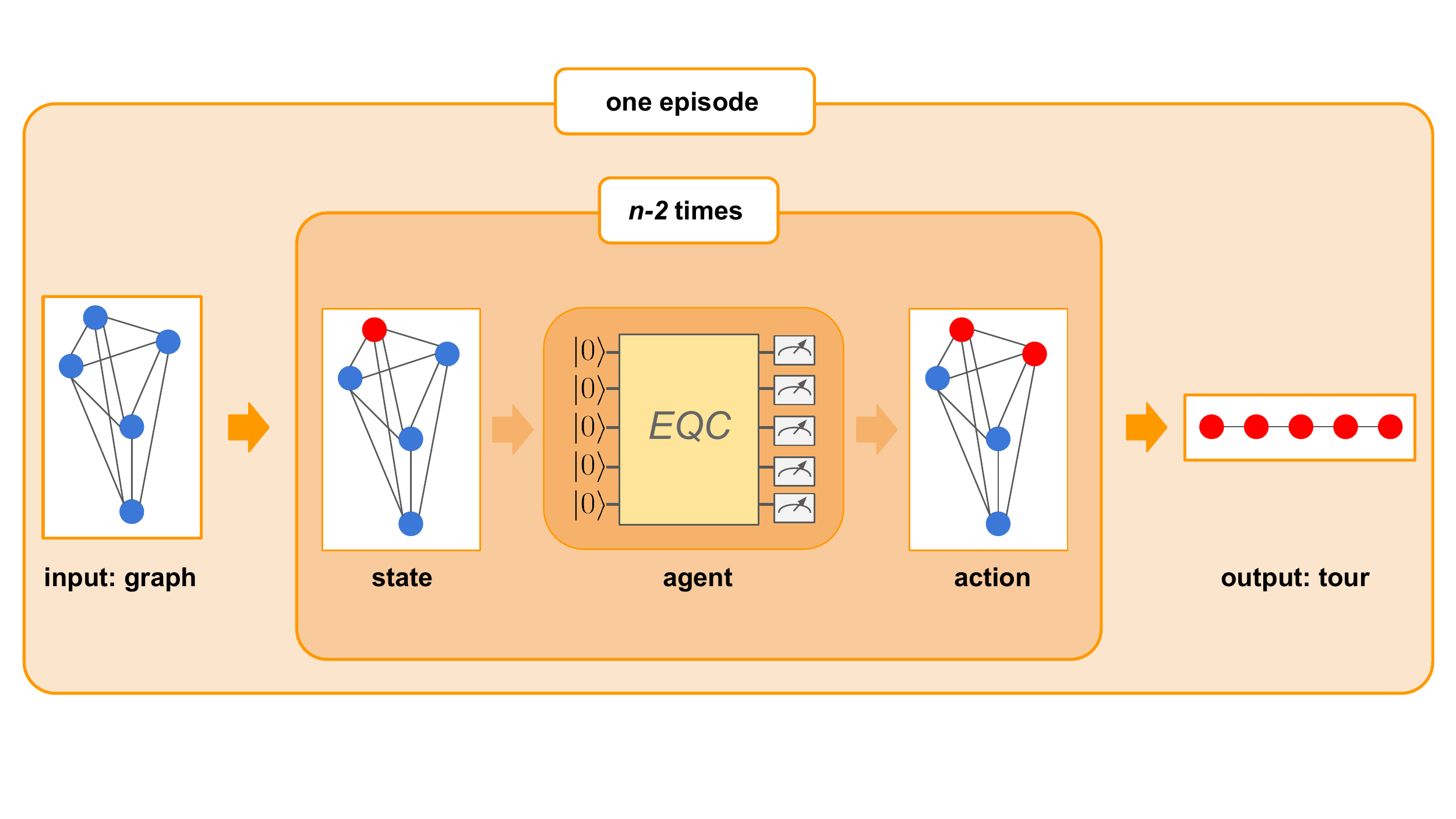}
\vspace*{-0.3in}
\caption{An illustration of one episode in the TSP environment. The agent receives a graph instance as input, where the first node is already added to the proposed tour (marked red), which can always be done without loss of generality. In each time step, the agent proposes which node should be added to the tour next. After the second-to-last node has been selected, the agent returns a proposed tour.}
\label{fig:qncorl_episodes}
\end{figure*}

Combinatorial optimization problems are ubiquitous, be it in transportation and logistics, electronics, or scheduling tasks. These types of problems have also been studied in computer science and mathematics for decades. Many interesting combinatorial optimization problems that are relevant in industry today are NP-hard, so that no general efficient solution is expected to exist. For this reason, heuristics have gained much popularity, as they often provide high-quality solutions to real-world instances of many NP-hard problems. However, good heuristics require domain expertise in their design and they have to be defined on a per-problem basis. To circumvent hand-crafting heuristic algorithms, machine learning approaches for solving combinatorial optimization problems have been studied. One line of research in this area investigates using NNs to learn algorithms for solving combinatorial optimization problems \cite{cappart2021combinatorial,bengio2021machine}, which is known as neural combinatorial optimization (NCO). Here, NNs learn to solve combinatorial optimization problems based on data, and can then be used to find approximate solutions to arbitrary instances of the same problem. First approaches in this direction used supervised learning to find approximate solutions based on NN techniques from natural language processing \cite{vinyals2015pointer}. A downside of the supervised approach is that it requires access to a large amount of training data in form of solved instances of the given problem, which requires solving many NP-hard instances of the problem to completion. At large problem sizes, this is a serious impediment for the practicability of this method. For this reason, reinforcement learning (RL) was introduced as a technique to train these heuristics. In RL, an agent does not learn based on a given data set, but by interacting with an environment and gathering information in form of rewards in a trial-and-error fashion. These RL-based approaches have been shown to successfully solve even instances of significant size in problems with a geometric structure like the convex hull problem \cite{dai2017learning}, chip placement \cite{mirhoseini2020chip} or the vehicle routing problem \cite{nazari2018reinforcement}. To implement NCO in this work we use a RL method called Q-learning following \cite{dai2017learning}, details of which are presented in \Cref{sec:related_work}.

Our goal is to use the ansatz described in \Cref{sec:ansatz_structure_equivariance} to train a model that, once trained, implements a heuristic to produce tours for previously unseen instances of the TSP. The TSP consists of finding a permutation of a set of cities such that
the resulting length of a tour visiting each city in this sequence is minimal. The heuristic takes as input an instance of the TSP problem in form of a weighted 2D Euclidean graph $\mathcal{G}(\mathcal{V}, \mathcal{E})$ with $n = |\mathcal{V}|$ vertices representing the cities and edge weights $\varepsilon_{ij} = \mathrm{d}(v_i, v_j)$, where $\mathrm{d}(v_i, v_j)$ is the Euclidean distance between nodes $v_i$ and $v_j$. Specifically, we are dealing with the symmetric TSP, where the edges in the graph are undirected. Given $\mathcal{G}$, the algorithm constructs a tour in $n-2$ steps. Starting from a given (fixed) node in the proposed tour $T_{t=1}$, in each step $t$ of the tour selection process the algorithm proposes the next node (city) in the tour. Once the second-before-last node has been added to the tour, the last one is also directly added, hence the tour selection process requires $n-2$ steps. This can also be viewed as the process of successively marking nodes in a graph as they are added to a tour. In order to refer to versions of the input graph at different time steps where the nodes that are already present in the tour are marked, we now define the annotated graph.

\begin{definition}[Annotated graph]\label{def:annotated_graph}
For a graph $\mathcal{G}(\mathcal{V}, \mathcal{E})$, we call $\mathcal{G}(\mathcal{V}, \mathcal{E}, \boldsymbol \alpha^{(t)})$ the annotated graph at time step $t$. The vector $\boldsymbol \alpha^{(t)} \in \{0, \pi\}^n$ specifies which nodes are already in the tour $T_t$ ($\alpha_i^{(t)} = 0$) and which nodes are still available for selection ($\alpha_i^{(t)} = \pi$).
\end{definition}
In each time step of an episode in the algorithm, the model is given an annotated graph as input. Based on the annotated graph, the model should select the next node to add to the partial tour $T_t$ at step $t$. The annotation can be used to partition the nodes $\mathcal{V}$ into the set of available nodes $\mathcal{V}_a = \{v_i | \alpha_i^{(t)} = \pi\}$ and the set of unavailable nodes $\mathcal{V}_u = \{v_i | \alpha_i^{(t)} = 0\}$. The node selection process can now be defined as follows.

\begin{definition}[Node selection]\label{def:node_selection}
Given an annotated graph $\mathcal{G}(\mathcal{V}, \mathcal{E}, \boldsymbol \alpha^{(t)})$, the node selection process consist of selecting nodes in a tour in a step-wise fashion. To add a node to the partial tour $T_t$, the next node is selected from the set of available nodes $\mathcal{V}_a$. The unavailable nodes $\mathcal{V}_u$ are ignored in this process.
\end{definition}
After $n-2$ steps, the model has produced a tour $T_n$. A depiction of this process can be found in \Cref{fig:qncorl_episodes}. To assess the quality of the generated tour, we compare the tour length $c(T_n)$ to the length of the optimal tour $c(T^*)$, where
\begin{equation}
    c(T) = \sum_{\{i,j\} \in \mathcal{E}_{T}} \varepsilon_{ij}
\end{equation}
is the sum of edge weights (distances) for all edges between the nodes in the tour, with $\mathcal{E}_{T} \subset \mathcal{E}$. We measure the quality of the generated tour in form of the approximation ratio
\begin{equation}\label{eq:approx_ratio}
    \frac{c(T_n)}{c(T^*)}.
\end{equation}
In order to perform Q-learning we need to define a reward function that provides feedback to the RL agent on the quality of its proposed tour. The rewards in this environment are defined by the difference in overall length of the partial tour $T_t$ at time step $t$, and upon addition of a given node $v_l$ at time step $t+1$:
\begin{equation}
    r(T_t, v_l) = - c(T_{(t+1, v_l)}) + c(T_t).
\end{equation}
Note that we use the negative of the cost as a reward, as a Q-learning agent will always select the action that leads to the maximum expected reward.

The learning process is defined in terms of a DQN algorithm, where the Q-function approximator is implemented in form of a PQC (which is described in detail in \Cref{sec:ansatz_structure_equivariance}). Here, we define the TSP in terms of an RL environment, where the set of states $\mathcal{S} = \{\mathcal{G}_i(\mathcal{V}, \mathcal{E}, \boldsymbol \alpha^{(t)}) ~\mathrm{for}~ i = 1, \dots, |\mathcal{X}| ~\mathrm{and}~ t = 1, \dots, n-1\}$ consists of all possible annotated graphs (i.e., all possible configurations of values of $\boldsymbol \alpha^{(t)}$) for each instance $i$ in the training set $\mathcal{X}$. This means that the number of states in this environment is $|\mathcal{S}| = 2^{n-1}|\mathcal{X}|$. The action that the agent is required to perform is selecting the next node in each step of the node selection process described in \Cref{def:node_selection}, so the action space $\mathcal{A}$ consists of a set of indices for all but the first node in each instance (as we always start from the first node in terms of the list of nodes we are presented with for each graph, so $\alpha_1^{(t)} = 0, ~\forall~ t$), and $|\mathcal{A}| = n-1$.

The Q-function approximator gets as input an annotated graph, and returns as output the index of the node that should next be added to the tour. Which index this is, is decided in terms of measuring an observable corresponding to each of the available nodes $\mathcal{V}_a$. Depending on the last node added to the partial tour, denoted as $v_{t-1}$, the observable for each available node $v_l$ is defined as
\begin{equation}\label{eq:observables}
    O_{v_l} = \varepsilon_{v_{t-1}, v_{l}} Z_{v_{t-1}}Z_{v_{l}}
\end{equation}
weighted by the edge weight $\varepsilon_{v_{t-1}, v_{l}}$, and the Q-value corresponding to each action is
\begin{multline}\label{eq:q_val}
    Q(\mathcal{G}_i(\mathcal{V}, \mathcal{E}, \boldsymbol \alpha^{(t)}), v_l) = \\
    \bra{\mathcal{E}, \boldsymbol \alpha^{(t)}, \boldsymbol \beta, \boldsymbol \gamma}_p O_{v_l} \ket{\mathcal{E}, \boldsymbol \alpha^{(t)}, \boldsymbol \beta, \boldsymbol \gamma}_p,
\end{multline}
where the exact form of $\ket{\mathcal{E}, \boldsymbol \alpha^{(t)}, \boldsymbol \beta, \boldsymbol \gamma}_p$ is described in \Cref{sec:ansatz_structure_equivariance}. The node that is added to the tour next is the one with the highest Q-value,
\begin{equation}
\mathrm{argmax}_{v_l} Q(\mathcal{G}_i(\mathcal{V}, \mathcal{E}, \boldsymbol \alpha^{(t)}), v_l).   
\end{equation}
All unavailable nodes $v_l \in \mathcal{V}_u$ are not included in the node selection process, so we manually set their Q-values to a large negative number to exclude them, e.g., \[Q(\mathcal{G}_i(\mathcal{V}, \mathcal{E}, \boldsymbol \alpha^{(t)}), v_l) = -10000 ~\forall~ v_l \in \mathcal{V}_u.\]

We also define a stopping criterion for our algorithm, which corresponds to the agent solving the TSP environment for a given instance size. As we aim at comparing the results of our algorithm to optimal solutions in this work, we have access to a labeled set of instances and define our stopping criterion based on these. However, note that the optimal solutions are not required for training, as a stopping criterion can also be defined in terms of number of episodes or other figures of merit that are not related to the optimal solution. In this work, the environment is considered as solved and training is stopped when the average approximation ratio of the past 100 iterations is $< 1.05$, where an approximation ratio of 1 means that the agent returns the optimal solution for the instances it was presented with in the past 100 episodes. We do not set the stopping criterion at optimality for two reasons: i) it is unlikely that the algorithm finds a parameter setting that universally produces the optimal tour for all training instances, and ii) we want to avoid overfitting on the training data set. If the agent does not fulfill the stopping criterion, the algorithm will run until a predefined number of episodes is reached. In our numerical results shown in \Cref{sec:num_results}, however, most agents do not reach the stopping criterion of having an average approximation ratio below 1.05, and run for the predefined number of episodes instead. Our goal is to generate a model that is, once fully trained, capable of solving previously unseen instances of the TSP. 

We showed in \Cref{sec:ansatz_structure_equivariance} that our ansatz of arbitrary depth is permutation equivariant. Now we proceed to show that the Q-values that are generated from measurements of this PQC, and the tour generation process as described in \Cref{sec:qgnn_architecture} are equivariant as well. While the equivariance of all components of an algorithm is not a pre-requisite to harness the advantage gained by an equivariant model, knowing which parts of our learning strategy fulfill this property provides additional insight for studying the performance of our model later. As we show that the whole node selection process is equivariant, we know that the algorithm will always generate the same tour for every possible permutation of the input graph for a fixed setting of parameters, given that the model underlying the tour generation process is equivariant. This is not necessarily true for a non-equivariant model, and simply by virtue of giving a permuted graph as input, the algorithm can potentially return a different tour.

\begin{theorem}[Equivariance of Q-values]\label{thm:equi_q_values}
Let $Q(\mathcal{G}(\mathcal{V}, \mathcal{E}, \boldsymbol \alpha), v_l) = Q(\mathcal{G}, v_l)$ be a Q-value as defined in \Cref{eq:q_val}, where we drop instance-specific sub- and superscripts for brevity. Let $\sigma$ be a permutation of $n = |\mathcal{V}|$ elements, where the $l$-th element corresponds to the $l$-th vertex $v_l$ and $\sigma_Q$ be a permutation that reorders the set of Q-values $Q (\mathcal{G}) = \{Q(\mathcal{G}, v_1), \dots, Q(\mathcal{G}, v_n)\}$ in correspondence to the reordering of the vertices by $\sigma$. Then the Q-values $Q(\mathcal{G})$ are permutation equivariant,
\begin{equation}
    Q(\mathcal{G}) = \sigma_Q Q(\mathcal{G}_{\sigma}),
\end{equation}
where $\mathcal{G}_{\sigma}$ is the permuted graph.
\end{theorem}

\begin{proof}
We know from \Cref{thm:ansatz_perm_equivariance} that the ansatz we use, and therefore the expectation values $\langle O_{v_l} \rangle$, are permutation equivariant. The Q-values are defined as $Q(\mathcal{G}, v_l) = \varepsilon_{ij} \langle O_{v_l} \rangle$ (see \Cref{eq:q_val}) and therefore additionally depend on the edge weights of the graph $\mathcal{G}$. The edge weights are computed according to the graph's adjacency matrix, and re-ordered under a permutation of the vertices and assigned to their corresponding permuted expectation values.
\end{proof}

As a second step, to show that all components of our algorithm are permutation equivariant, it remains to show that the tours that our model produces as described in \Cref{sec:qgnn_architecture} are also permutation equivariant.

\begin{corollary}[Equivariance of tours]
Let $T(\mathcal{G}, \boldsymbol \beta, \boldsymbol \gamma, v_0)$ be a tour generated by a permutation equivariant agent implemented with a PQC as defined in \Cref{eq:ansatz} and Q-values as defined in \Cref{eq:q_val}, for a fixed set of parameters $\boldsymbol \beta, \boldsymbol \gamma$ and a given start node $v_0$, where a tour is a cycle over all vertices $v_l \in \mathcal{V}$ that contains each vertex exactly once. Let $\sigma$ be a permutation of the vertices $\mathcal{V}$, and $\sigma_T$ a permutation that reorders the vertices in the tour accordingly. Then the output tour is permutation equivariant,
\begin{equation}
    T(\mathcal{G}, \boldsymbol \beta, \boldsymbol \gamma, v_0) = \sigma_T T(\mathcal{G}_{\sigma}, \boldsymbol \beta, \boldsymbol \gamma, v_{\sigma(0)}).
\end{equation}
\end{corollary}

\begin{proof}
We have shown in \Cref{thm:equi_q_values} that the Q-values of our model are permutation equivariant, meaning that a permutation of vertices results in a reordering of Q-values to different indices. Action selection is done by $v_{t+1} = \mathrm{argmax}_v Q(\mathcal{G}_i^{(t)}, v)$, and the node at the index corresponding to the largest Q-value is chosen. To generate a tour, the agent starts at a given node $v_0$ and sequentially selects the following $n-1$ vertices. Upon a permutation of the input graph, the tour now starts at another node index $v_{\sigma(0)}$. Each step in the selection process can now be seen w.r.t.\ the original graph $\mathcal{G}$ and the permuted graph $\mathcal{G}_{\sigma}$. As we have shown in \Cref{thm:ansatz_perm_equivariance}, equivariance of the model holds for arbitrary input graphs, so in particular it holds for each $\mathcal{G}$ and $\mathcal{G}_{\sigma}$ in the action selection process, and the output tour under the permuted graph is equal to the output tour under the original graph up to a renaming of the vertices.
\end{proof}

\subsection{Analysis of expressivity}\label{sec:opt_containing_proof}

In this section, we analyze under which conditions there exists a setting of $\beta, \gamma$ for a given graph instance $\mathcal{G}_i$ for our ansatz at depth one that can produce the optimal tour for this instance. Note that this does not show anything about constructing the optimal tour for a number of instances simultaneously with this set of parameters, or how easy it is to find any of these sets of parameters. Those questions are beyond the scope of this work. The capability to produce optimal tours at any depth for individual instances is of interest because first, we do not expect that the model can find a set of parameters that is close-to-optimal for a large number of instances if it is not expressive enough to contain a parameter setting that is optimal for individual instances. Second, the goal of a ML model is always to find similarities within the training data that can be used to generalize well on the given learning task, so the ability to find optimal solutions on individual instances is beneficial for the goal of generalizing on a larger set of instances. Additionally, how well the model generalizes also depends on the specific instances and the parameter optimization routine, and therefore it is hard to make formal statements about the general case where we find one universal set of parameters that produces the optimal solution for arbitrary sets of instances.

For our model at $p=1$, we can compute the analytic form of the expectation values of our circuit as defined in \Cref{eq:observables} and \Cref{eq:q_val} as the following, by a similar derivation as in \cite{ozaeta2022expectation},
\begin{multline}\label{eq:anal_expectation}
    \langle O_{v_l} \rangle = \varepsilon_{v_{t-1}, v_l} \cdot \mathrm{sin}(\beta \pi) ~\mathrm{sin}(\varepsilon_{v_{t-1}, v_l} \gamma)\\
    \cdot \prod_{\underset{k \neq v_{t-1}}{(v_l,k) \in \mathcal{E}}} \mathrm{cos}(\varepsilon_{v_l, k} \gamma),
\end{multline}
where $v_{t-1}$ is the last node in the partial tour and $v_l$ is the candidate node. Note that due to the specific setup of node features used in our work where the contributions of nodes already present in the tour are turned off, these expectation values are simpler than those given for Ising-type Hamiltonians in \cite{ozaeta2022expectation}. For a learning task where contributions of all nodes are present in every step, the expectation values of the EQC will be the same as those for Ising Hamiltonians without local fields given in \cite{ozaeta2022expectation}, with the additional node features $\boldsymbol \alpha$.  Due to this structural similarity to the ansatz used in the QAOA, results on the hardness to give an analytic form of these expectation values at $p > 1$ also transfer to our model. Even at depth $p=2$ analytic expressions can only be given for certain types of graphs \cite{marwaha2021local,szegedy2019qaoa}, and everything beyond this quickly becomes too complex. For this reason, we can only make statements for $p=1$ in this work. 

In order to generate an arbitrary tour of our choice, in particular also the optimal tour, it suffices to guarantee that for a suitable choice of (fixed) $\gamma$, at each step in the node selection process the edge we want to add next to the partial tour has highest expectation. One way we can do this is by controlling the signs of each sine and cosine term in \Cref{eq:anal_expectation} such that only the expectation values corresponding to edges that we want to select are positive, and all others are negative.

To understand whether this is possible, we can leverage known results about the expressivity of the sine function. For any rationally independent set of $\{x_1, ..., x_n\}$ with labels $y_i \in (-1, 1)$, the sine function can approximate these points to arbitrary precision $\epsilon$ as shown in \cite{harman2015sin}, i.e., there exists an $\omega$ s.t.\
\begin{equation}
    |\mathrm{sin}(\omega x_i) - y_i| < \epsilon ~\mathrm{for}~ i=1, \dots, n.
\end{equation}

In general, the edge weights of graphs that represent TSP instances are not rationally independent. (The real numbers $x_1, \dots, x_n$ are said to be rationally independent if no integers $k_1, \dots, k_n$ exist such that $x_1 k_1 + \dots + x_n k_n = 0$, besides the trivial solution $k_i = 0 ~\forall~ k$. Rational independence also implies the points are not rational numbers, so they are also not numbers normally represented by a computer.) However, in principle they can easily be made rationally independent by adding a finite perturbation $\epsilon'_i$ to each edge weight. The results in \cite{harman2015sin} imply that almost any set of points $x_1, \dots, x_n$ with $0 < x_i < 1$ is rationally independent, so we can choose $\epsilon'_i$ to be drawn uniformly at random from $(0, \epsilon_{\mathrm{max}}]$. As long as these perturbations are applied to the edge weights in a way that does not change the optimal tour, as could be done by ensuring that $\epsilon_{\mathrm{max}}$ is small enough so that the proportions between edge weights are preserved, we can use this perturbed version of the graph to infer the optimal tour. (Such an $\epsilon_{\mathrm{max}}$ can be computed efficiently.) In this way we can guarantee that the ansatz at depth one can produce arbitrary labelings of our edges, which in turn let us produce expectation values such that only the ones that correspond to edges in the tour of our choice will have positive values. We note that in the analysis we assume real-valued (irrational) perturbations, which of course cannot be represented in the computer. However, by using the results of \cite{harman2015sin} and approximating $\pm 1$ within a small epsilon, we can get a robust statement where finite precision suffices. For completeness, we provide a proof for this case in the supplementary material. However, we point out that the parameter $\omega$ that leads to the construction of the optimal tour can in principle be arbitrarily large and hard to find.  We do not go deeper into this discussion since in fact we do not want to rely on this proof of optimality as a guiding explanation of how the algorithm works.

The reason for this is that in some way, this proof of optimality works despite the presence of the TSP graph and not because of it. This is similar in vein to universality results for QAOA-type circuits, where it can be shown that for very specific types of Hamiltonians, alternating applications of the cost and mixer Hamiltonian leads to quantum computationally universal dynamics, i.e.\ it can reach all unitaries to arbitrary precision \cite{lloyd2018quantum,morales2020universality}, but these Hamiltonians are not related to any of the combinatorial optimization problems that were studied in the context of the QAOA. While these results provide valuable insight into the expressivity of the models, in our case they do not inform us about the possibility of a quantum advantage on the learning problem that we study in this work. In particular, we do not know from these results whether the EQC utilizes the information provided by the graph features in a way in which the algorithm benefits from the quantumness of the model, at depth one or otherwise. As it is known that the QAOA applied to ground state finding benefits from interference effects, investigating whether similar results hold for our algorithm is an interesting question that we leave for future work.

Additionally, we note that high expressivity alone does not necessarily lead to a good model, and may even lead to issues in training as the well-studied phenomenon of barren plateaus \cite{mcclean2018barren}, or a susceptibility to overfitting on the training data. In practice, the best models are those that strike a balance between being expressive enough, and also restricting the search space of the model in a way that suits the given training data. Studying and designing models that have this balance is exactly the goal of geometric learning, and the equivariance we have proven for our model is a helpful geometric prior for learning tasks on graphs.

\subsection{Numerical results}\label{sec:num_results}

\begin{figure*}
  \subfloat{\includegraphics[scale=0.66]{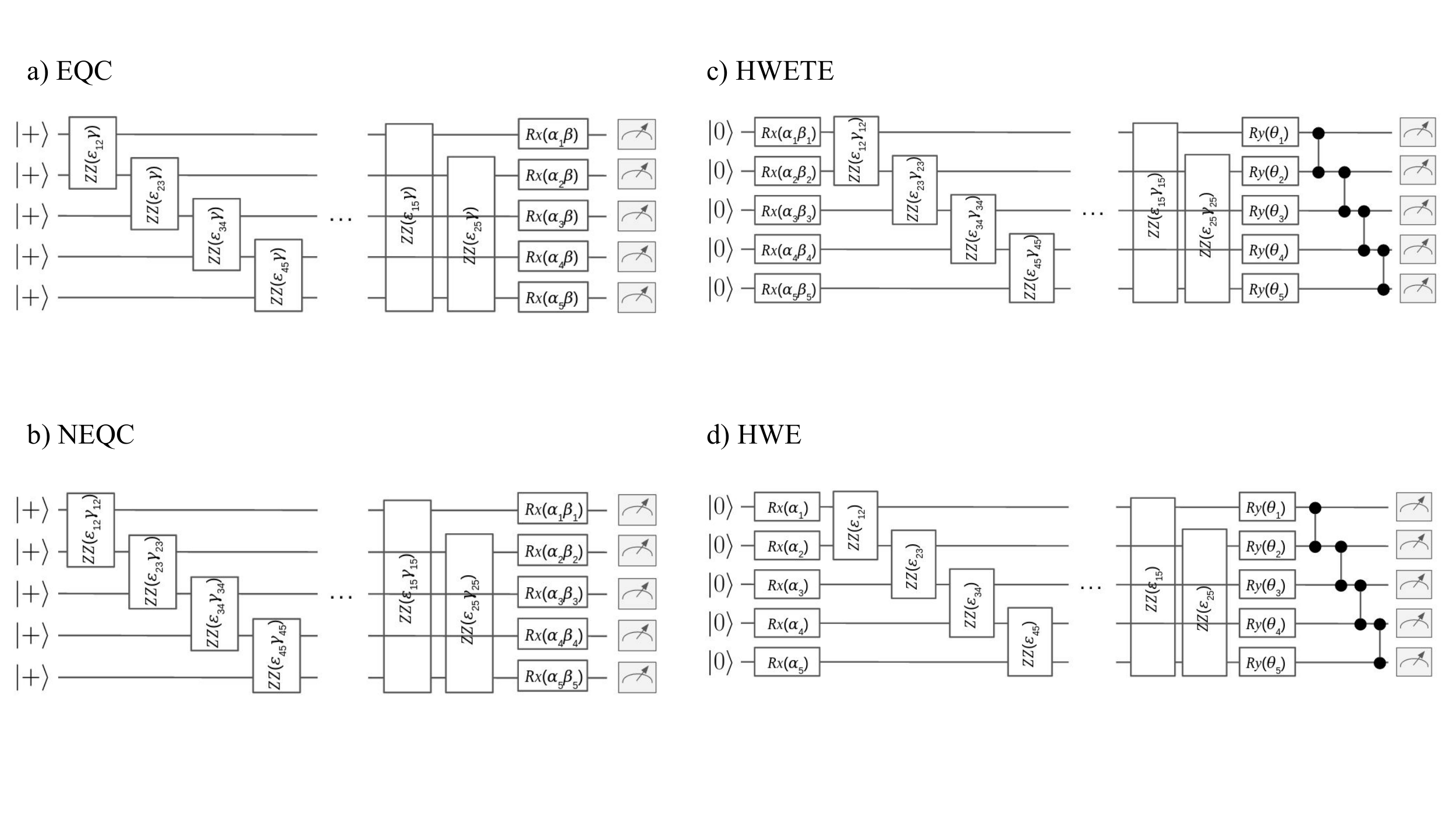}}
  \vspace*{-0.4in}
 \caption{One layer of each of the circuits studied in this work. a) The EQC with two trainable parameters $\beta, \gamma$ per layer. b) The same set of gates as in the EQC, but we break equivariance by introducing one individual free parameter per gate (denoted NEQC). c) Similar to the NEQC, but we start from the all-zero state and add a final layer of trainable one-qubit gates and a ladder of CZ-gates (denoted hardware-efficient with trainable encoding, HWETE). d) Same as the HWETE, but only the single-qubit Y-rotation parameters are trained (denoted HWE).}
\label{fig:circuits}
\end{figure*}

\begin{figure*}
  \subfloat{\includegraphics[scale=0.5]{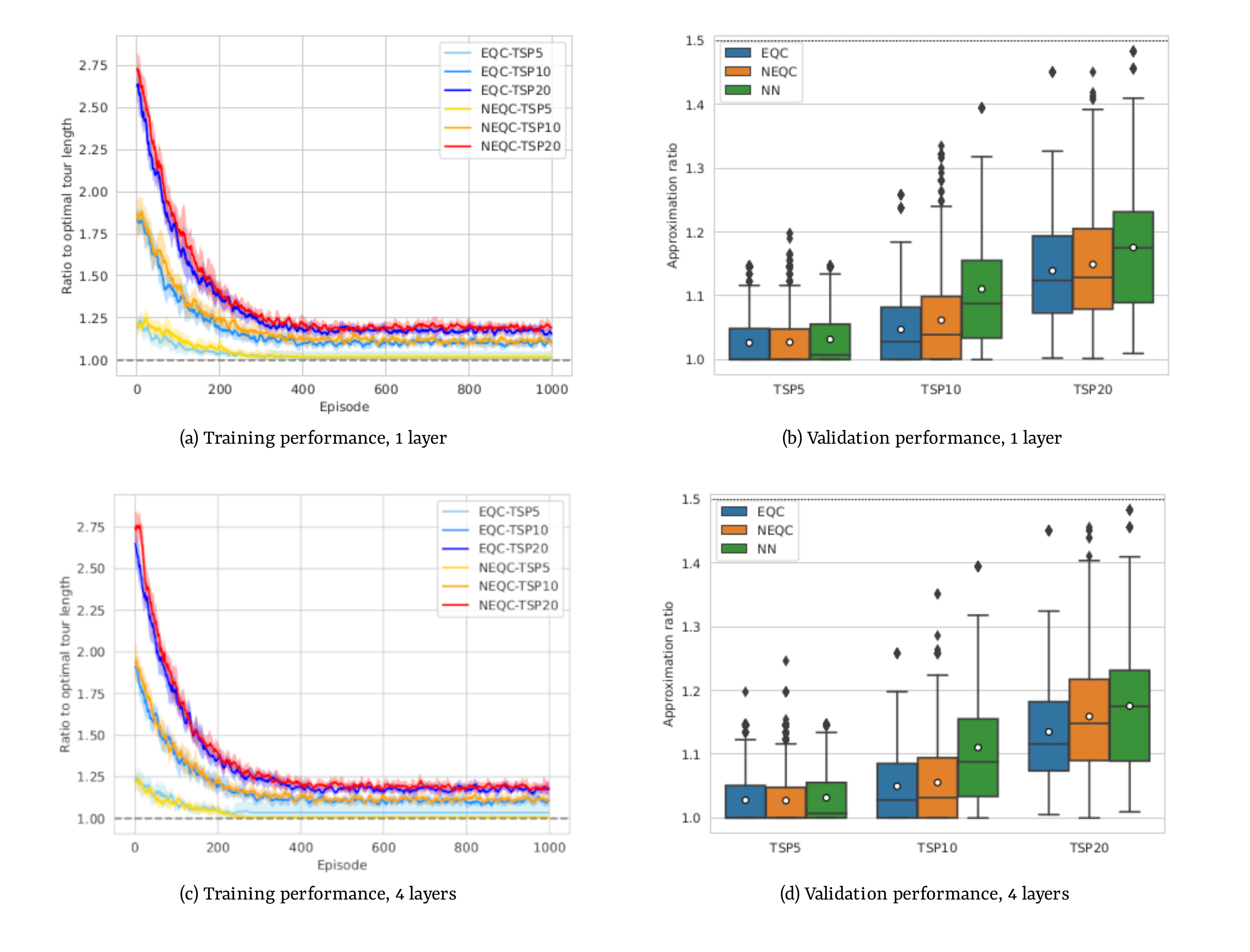}}
  \vspace*{-0.2in}
 \caption{Comparison between the EQC and its non-equivariant version (NEQC). Comparison is done in terms of approximation ratio (lower is better) of ten trained models on a set of 100 previously unseen TSP instances for each instance size. Shaded areas in left-hand side figures show standard error of the mean. The boxes on the right-hand side show the upper quartile, median and lower quartile for each model configuration, the whiskers of the boxes extend to 1.5 times the interquartile range, and the black diamonds denote outliers. We additionally show the means of each box as white circles.  In the NEQC each gate is parametrized separately but the ansatz structure is otherwise identical to the EQC, as described in \Cref{sec:num_results}. Results are shown on TSP instances with 5, 10 and 20 cities (TSP5, TSP10 and TSP20, respectively). To provide a classical baseline, we also show results for the nearest-neighbor heuristic (NN). a) and b) show the training and validation performance for both ansatzes with one layer, while c) and d) show the same for four layers. The dashed, grey line on the left-hand side figures denotes optimal performance. The dotted, black line on the right-hand side figures denotes the upper bound of the Christofides algorithm, a popular classical approximation algorithm that is guaranteed to find a solution that is at most 1.5 times as long as the optimal tour. Figures a) and c) show the running average over the last ten episodes.}
\label{fig:comparison_eq_neq_nn}
\end{figure*}

After proving that our model is equivariant under node permutations and analytically studying the expressivity of our ansatz, we now numerically study the training and validation performance of this model on TSP instances of varying size in a NCO context. The training data set that we use is taken from \cite{vinyals2015pointer}, where the authors propose a novel classical attention approach and evaluate it on a number of geometric learning tasks. We note that we have re-computed the optimal tours for all instances that we use, as the data set uploaded by the authors of \cite{vinyals2015pointer} erroneously contains sub-optimal solutions. (This was confirmed with the authors, but at the time of writing of this work their repository has not been updated with the correct solutions.) To compute optimal solutions for the TSP instances with 10 and 20 cities we used the library Python-TSP \cite{py_tsp}.

We evaluate the performance of the EQC on TSP instances with 5, 10 and 20 cities (corresponding to 5, 10, and 20 qubits, respectively). As described in \Cref{sec:qgnn_architecture}, the environment is considered as solved by an agent when the running average of the approximation ratio over the past 100 episodes is less than 1.05. Otherwise, each agent will run until it reaches the maximum number of episodes, that we set to be 5000 for all agents. Note that this is merely a convenience to shorten the overall training times, as we have access to the optimal solutions of our training instances. In a realistic scenario where one does not have access to optimal solutions, the algorithm would simply run for a fixed number of episodes or until another convergence criterion is met. When evaluating the final average approximation ratios, we always use the parameter setting that was stored in the final episode, regardless of the final training error. When variations in training lead to a slightly worse performance than what was achieved before, we still use the final parameter setting. We do this because as noted above, in a realistic scenario one does not have knowledge about the ratio to the optimal solutions during training. Unless otherwise stated, all models are trained on 100 training instances and evaluated on 100 validation instances.

\begin{figure*}
  \subfloat{\includegraphics[scale=0.5]{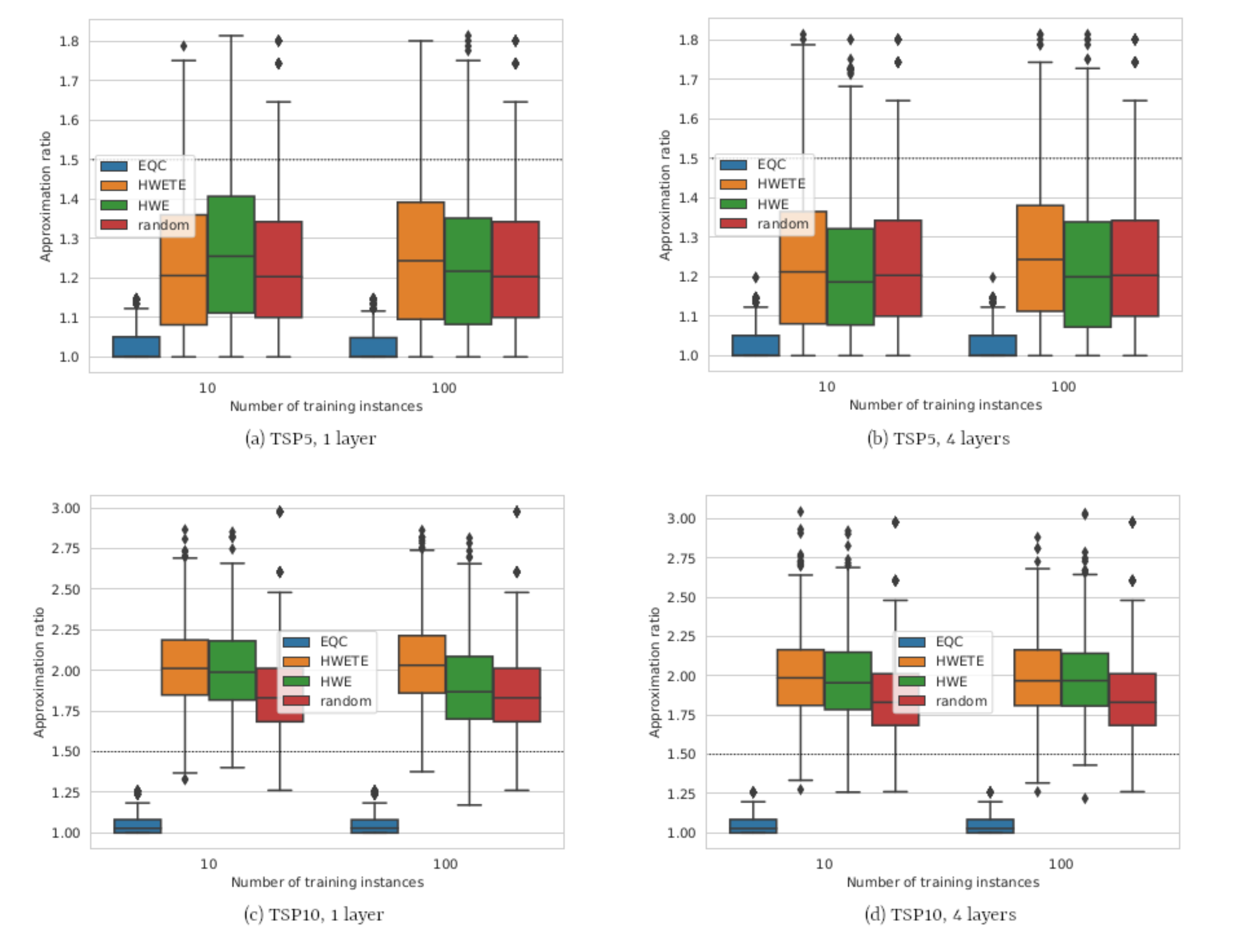}}
  \vspace*{-0.2in}
 \caption{Comparison between EQC and two hardware-efficient ansatzes where we gradually break the equivariance of the original ansatz. We show results for TSP instances with five (TSP5) and ten (TSP10) cities in a) and b), and c) and d), respectively. Models were trained on 10 and 100 instances, and with one and four layers. Each box is computed over results for ten agents. The hardware-efficient ansatz with trainable embedding (HWETE) consists of trainable graph encoding layers as those in the EQC, with an additional variational part in each layer that consists of parametrized single-qubit Y-gates and a ladder of CZ-gates. The HWE ansatz is the same as the HWETE, but where the graph-embedding part is not trainable and only the Y-gates in each layer are trained. We also show approximation ratios of a random algorithm, where a random tour is picked as the solution to each instance. The dotted, black lines denote the upper bound of the Christofides algorithm. We see that the HWE ansatzes perform extremely badly and barely outperform picking random tours only in some cases.}
\label{fig:comparison_hwe}
\end{figure*}

As we are interested in the performance benefits that we gain by using an ansatz that respects an important graph symmetry, we compare our model to versions of the same ansatz where we gradually break the equivariance property. We start with the simplest case, were the circuit structure is still the same as for the EQC, but instead of having one $\beta_l, \gamma_l$ in each layer, every X- and ZZ-gate is individually parametrized. As these parameters are now tied directly to certain one- and two-qubits gates, e.g.\ an edge between qubits one and two, they will not change location upon a graph permutation and therefore break equivariance. We call this the non-equivariant quantum circuit (NEQC). To go one step further, we take the NEQC and add a variational part to each layer that is completely unrelated to the graph structure: namely a hardware-efficient layer that consists of parametrized Y-rotations and a ladder of CZ-gates. In this ansatz, we have a division between a data encoding part and a variational part, as is often done in QML. To be closer to standard types of ansatzes often used in QML, we also omit the initial layer of H-gates here and start from the all-zero state, which requires us to switch the order of X- and ZZ-gates (however, in practice it did not make a difference whether we started from the all-zero or uniform superposition state in the learning task that we study). We denote this the hardware-efficient with trainable embedding (HWETE) ansatz. Finally, we study a third ansatz, where we take the HWETE and now only train the Y-rotation gates, and the graph-embedding part of the circuit only serves as a data encoding step. We call this simply the hardware-efficient (HWE) ansatz. A depiction of all ansatzes can be seen in \Cref{fig:circuits}.

We start by comparing the EQC to the NEQC on TSP instances with 5, 10 and 20 cities. We show the training and validation results in \Cref{fig:comparison_eq_neq_nn}. To evaluate the performance of the models that we study, we compute the ratio to the optimal tour length as shown in \Cref{eq:approx_ratio}, as the instances that we can simulate the circuits for are small enough to allow computing optimal tours for. For reference, the authors of \cite{vinyals2015pointer}, who generated the training instances that we use, stop comparing to optimal solutions at $n=20$ as it becomes extremely costly to find optimal tours from thereon out. To provide an additional classical baseline, we also show results for the nearest-neighbor heuristic. This heuristic starts at a random node and selects the closest neighboring node in each step to generate the final tour. The nearest-neighbor algorithm finds a solution quickly also for instances with increasing size, but there is no guarantee that this tour is close to the optimal one. However, as we know the optimal tours for all instances, the nearest-neighbor heuristic provides an easy to understand classical baseline that we can use. Additionally, we add the upper bound given by one of the most widely used approximation algorithms for the TSP (as implemented e.g.\ in Google OR-Tools): the Christofides algorithm. This algorithm is guaranteed to find a tour that is at most 1.5 times as long as the optimal tour \cite{christofides1976worst}. In the case where any of our models produces validation results that are on average above this upper bound of the Christofides algorithm, we consider it failed, as it is more efficient to use a polynomial approximation algorithm for these instances. However, we stress that this upper bound can only serve to inform us about the failure of our algorithms and not their success, as in practice the Christofides algorithm often achieves much better results than those given by the upper bound. We also note that both the Christofides and nearest-neighbor algorithms are provided here to assure that our algorithm produces reasonable results, and not to show that our algorithm outperforms classical methods as this is not the topic of the present manuscript.  The bound is shown as a dotted black line in \Cref{fig:comparison_eq_neq_nn} and \Cref{fig:comparison_hwe}.

Geometric learning models are expected to be more data-efficient than their unstructured counterparts, as they respect certain symmetries in the training data. This means that when a number of symmetric instances are present in the training or validation data, the effective size of these data sets is decreased. This usually translates into models that are more resource-efficient in training, e.g.\, by requiring fewer parameters or fewer training samples. In our comparison of the EQC and the NEQC, we fix the number of training samples and compare the different models in terms of circuit depth and number of parameters to achieve a certain validation error and expect that the EQC will need fewer layers to achieve the same validation performance as the NEQC. This comparison can be seen in \Cref{fig:comparison_eq_neq_nn}. In \Cref{fig:comparison_eq_neq_nn} a) and b), we show the training and validation performance of both ansatzes at depth one. For instances with five cities, both ansatzes perform almost identically on the validation set, where the NEQC performs worse on a few validation instances. As the instance size increases, the gap between EQC and NEQC becomes bigger. We see that even though the two ansatzes are structurally identical, the specific type of parametrizations we choose and the properties of both ansatzes that result from this make a noticeable difference in performance. While the EQC at depth one has only two parameters per layer regardless of instance size, the NEQC's number of parameters per layer depends on the number of nodes and edges in the graph. Despite having much fewer parameters, the EQC still outperforms the NEQC on instances of all sizes. Increasing the depth of the circuits also does not change this. In \Cref{fig:comparison_eq_neq_nn} c) and d) we see that at a depth of four, the EQC still beats the NEQC. The latter's validation performance even slightly decreases with more layers, which is likely due to the increased complexity of the optimization task, as the number of trainable parameters per layer is $\frac{(n-1)n}{2} + n$, which for the 20-city instances means 840 trainable parameters at depth four (compared to 8 parameters in case of the EQC). This shows that at a fraction of the number of trainable parameters, the EQC is competitive with its non-equivariant counterpart even though the underlying structure of both circuits is identical. Compared to the classical nearest-neighbor heuristic, both ansatzes perform well and beat it at all instance sizes, and both ansatzes are also below the approximation ratio upper bound given by the Christofides algorithm on all validation instances.  The box plots in \Cref{fig:comparison_eq_neq_nn} show a comparison of the EQC and NEQC in terms of the quartiles of the approximation ratios on the validation set. As it is hard to infer statistical significance of results directly from the box plots, especially when the distributions of data points are not very far apart, we additionally plot the means of the distributions and their standard error, and compute p-values based on a t-test to give more insight on the comparison of these two models in the supplementary material. To show statistical significance of the comparison of the EQC and NEQC, we perform a two-sample t-test with the null-hypothesis that the averages of the two distributions are the same, as is common in statistical analysis, and compute p-values based on this. The p-values confirm that there is indeed a statistical significance in the comparison between models for the 10- and 20-city instances, and that we can be more certain about the significance as we scale up the instance size. The average approximation ratios in case of the 5-city instances are roughly the same, as we can expect due to the fact that there exist only 12 permutations of the TSP graphs of this size. However, even for these small instances the EQC achieves the same result with fewer parameters, namely 2 per layer instead of the 15 per layer required in the NEQC. 

Next, we compare the EQC to ansatzes in which we introduce additional variational components that are completely unrelated to the training data structure, as described above. We show results for the HWETE and the HWE ansatz in \Cref{fig:comparison_hwe}. To our own surprise, we did not manage to get satisfactory results with either of those two ansatzes, especially at larger instances, despite an intensive hyperparameter search. Even the HWETE, which is basically identical to the NEQC with additional trainable parameters in each layer, failed to show any significant performance. To gauge how badly those two ansatzes perform, we also show results for an algorithm that selects a random tour for each validation instance in \Cref{fig:comparison_hwe}. In this figure, we show results for TSP instances with five and ten cities for both ansatzes with one and four layers, respectively. Additionally, we show how the validation performance changes when the models are trained with either a training data set consisting of 10 or 100 instances, in the hopes of seeing improved performance as the size of the training set increases. We see that in neither configuration, the HWETE or HWE outperform the Christofides upper bound on all validation instances. Additionally, in almost all cases those two ansatzes do not even outperform the random algorithm. This example shows that in a complex learning scenario, where the number of permutations of each input instance grows combinatorially with instance size and the number of states in the RL environment grows exponentially with the number of instances, a simple hardware-efficient ansatz will fail even when the data encoding part of the PQC is motivated by the problem data structure. While increasing the size of the training set and/or the number of layers in the circuit seems to provide small advantages in some cases, it also leads to a decrease in performance in others. On the other hand, the EQC is mostly agnostic to changes in the number of layers or the training data size. Overall, we see that the closer the ansatz is to an equivariant configuration, the better it performs, and picking ansatzes that respect symmetries inherent to the problem at hand is the key to success in this graph-based learning task.

In \Cref{sec:ansatz_structure_equivariance} we have also pointed out that the EQC is structurally related to the ansatz used in the QAOA. The main difference in solving instances of the TSP with the NCO approach used in our work and solving it with the QAOA lies in the way in which the problem is encoded in the ansatz, and in the quantity that is used to compute the objective function value for parameter optimization. We give a detailed description of how the TSP is formulated in terms of a problem Hamiltonian suitable for the QAOA and how parameters are optimized in \Cref{sec:tsp_qaoa}. As the QAOA is arguably the most explored variational quantum optimization algorithm at the time of writing, and due to the structural similarity between the EQC and the QAOA's ansatz, we also compare these two approaches on TSP instances with five cities. 

\begin{figure}[t]
\includegraphics[width=0.5\textwidth]{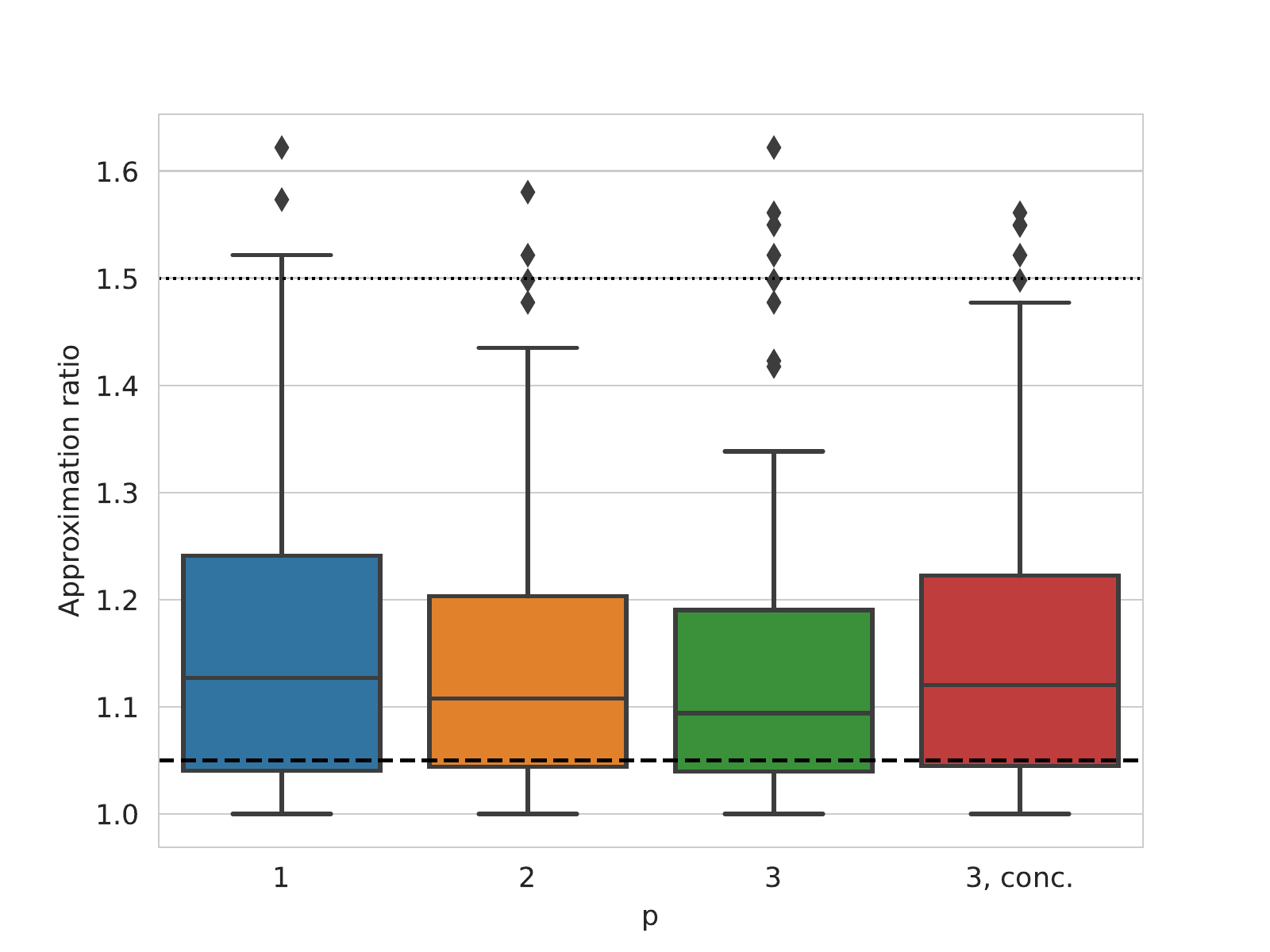}
\caption{Approximation ratio of QAOA on TSP instances with five cities up to depth three. The boxes show the upper quartile, median and lower quartile for each model configuration, the whiskers of the boxes extend to 1.5 times the interquartile range, and the black diamonds denote outliers. Dashed black line denotes average final performance of the EQC at depth one during the last 100 iterations of training on the same instances. Last box shows the results for the best parameters found for one instance at $p=3$ applied to all training instances, following a parameter concentration argument. The dotted, black line denotes the upper bound of the Christofides algorithm.}
\label{fig:tsp5_qaoa}
\end{figure}
We see in \Cref{fig:tsp5_qaoa} that already on these small instances, the QAOA requires significantly deep circuits to achieve good results, that may be out of reach in a noisy near-term setting. The EQC on the other hand i) uses a number of qubits that scales linearly with the number of nodes of the input graph as opposed to the $n^2$ number of variables required for QAOA, and ii) already shows good performance at depth one for instances with up to 20 cities. In addition to optimizing QAOA parameters for each instance individually, we also show results of applying one set of parameters that performed well on one instance at depth three, on other instances of the same problem following the parameter concentration argument given in \cite{brandao2018concentration} and described in more detail in \Cref{sec:tsp_qaoa}. While we find that parameters seem to transfer well to other instances of the same problem in case of the TSP, the overall performance of the QAOA is still much worse than that of the EQC.

\section{Discussion}

After providing analytic insight on the expressivity of our ansatz, we have numerically investigated the performance of our EQC model on TSP instances with 5, 10, and 20 cities (corresponding to 5, 10, 20 qubits respectively), and compared them to other types of ansatzes that do not respect any graph symmetries. To get a fair comparison, we designed PQCs that gradually break the equivariance property of the EQC and assessed their performance. We find that ansatzes that contain structures that are completely unrelated to the input data structure are extremely hard to train for this learning task where the size of the state space scales exponentially in the number of input nodes of the graph. Despite much effort and hyperparameter optimization, we did not manage to get satisfactory results with these ansatzes. The EQC on the other hand works almost out-of-the box, and achieves good generalization performance with minimal hyperparameter tuning and relatively few trainable parameters. We have also compared using the EQC in a neural combinatorial optimization scheme with the QAOA, and find that even on TSP instances with only five cities the NCO approach significantly outperforms the QAOA. In addition to training the QAOA parameters for every instance individually, we have also investigated the performance in light of known parameter concentration results that state that in some cases, parameters found on one instance perform well on average for other instances of the same problem. While this is true in the case of the TSP instances we investigate here as well, the overall performance is still worse than that of the EQC.

Comparing our algorithm to the QAOA is also interesting from a different perspective. In \Cref{sec:ansatz_structure_equivariance} we have seen that our ansatz can be regarded as a special case of a QAOA-type ansatz, where instead of encoding a problem Hamiltonian we encode a graph instance directly, and in case of the specific formulation of the TSP used in this work, include mixing terms only for a problem-dependent subset of qubits. This lets us derive an exact formulation of the expectation values of our model at depth one from those of the QAOA given in \cite{ozaeta2022expectation}. For the QAOA, it is known that in the limit of infinite depth, it can find the ground state of the problem Hamiltonian and therefore the optimal solution to a given combinatorial optimization problem \cite{farhi2014quantum}. Additionally, it has been shown that even at low depth, the probability distributions generated by QAOA-type circuits are hard to sample from classically \cite{farhi2016quantum}. These results give a clear motivation of why using a quantum model in these settings can provide a potential advantage. While our model is structurally almost identical to that of the QAOA, in our case the potential for advantage is less clear. We saw in \Cref{eq:anal_expectation} that at depth one, in each step the expectation value of each edge that we consider to be selected consists of i) a term corresponding to the edge between the last added node and the candidate node, and ii) all outgoing edges from the candidate node. So our model considers the one-step neighborhood of each candidate node at depth one. In the case of the TSP it is not clear whether this can provide a quantum advantage for the learning task as specified in \Cref{sec:qgnn_architecture}. In terms of QAOA, it was shown that in order to find optimal solutions, the algorithm has to ``see the whole graph" \cite{farhi2020quantum}, meaning that all edges in the graph contribute to the expectation values used to minimize the energy. To alleviate this strong requirement on depth, a recursive version of the QAOA (RQAOA) was introduced in \cite{bravyi2020obstacles}. It works by iteratively eliminating variables in the problem graph based on their correlation, and thereby gradually reducing the problem to a smaller instance that can be solved efficiently by a classical algorithm,  e.g.\ by brute-force search. The authors of \cite{bravyi2020obstacles} show that the depth-one RQAOA outperforms QAOA with constant depth $p$, and that RQAOA achieves an approximation ratio of one for a family of Ising Hamiltonians. 

The node selection process performed by our algorithm with the EQC used as the ansatz is similar to the variable elimination process in the RQAOA, where instead of merging edges, the mixer terms for nodes that have already been selected are turned off, therefore effectively turning expectation values of edges corresponding to unavailable nodes to zero. Furthermore, the specific setup of weighted $Z_iZ_j$-correlations (see \Cref{eq:observables}) that we measure to compute Q-values in our RL scheme to solve TSP instances are of the same form as those in the Hamiltonian for the weighted MaxCut problem, 
$$H_{\mathrm{MaxCut}} = -\frac{1}{2} \sum_{ij} w_{ij} (1 - Z_i Z_j).$$
The MaxCut problem and its weighted variant have been studied in depth in the context of the QAOA, and it has been shown that it performs well on certain instances of graphs for this task \cite{farhi2014quantum,zhou2020quantum,shaydulin2022parameter,brandao2018concentration}. While the TSP and weighted MaxCut are clearly very different problems, the similarity between our algorithm and the RQAOA raises the interesting question whether the mechanisms underlying the successful performance of both models in those two learning tasks are related. Based on this, one may ask the broader question of whether QAOA-type ansatzes implement a favorable bias for hybrid quantum-classical optimization algorithms on weighted graphs, like the RQAOA or the quantum NCO scheme in this work. Specifically, by relating the mechanism underlying the variable elimination procedure in RQAOA, which eliminates variables based on their largest (anti-)correlation in terms of $Z_iZ_j$ operators, to the node selection process in our algorithm that solves TSP instances, we can establish a connection between the EQC and known results for the (R)QAOA on weighted MaxCut. It is an interesting question whether results that establish a quantum advantage of the QAOA can be related to the EQC in a NCO context as we present here, and we leave this question for future work.

\section{Methods}
\label{sec:related_work}

\subsection{Geometric learning - quantum and classical}

Learning approaches that utilize geometric properties of a given problem have lead to major successes in the field of ML, such as AlphaFold for the complex task of protein folding \cite{jumper2021highly,tunyasuvunakool2021highly} and have become an increasingly popular research field over the past few years. Arguably the prime example of a successful geometric model is the convolutional NN (CNN), which has been developed at the end of the 20th century in an effort to enable efficient training of image recognition models \cite{lecun1995convolutional}. Since then, it has been shown that one of the main reasons that CNNs are so effective is that they are translation invariant: if an object in a given input image is shifted by some amount, the model will still ``recognize" it as the same object and thus effectively requires fewer training data \cite{bronstein2021geometric}. While CNNs are the standard architecture used for images, symmetry-preserving architectures have also been developed for time-series data in the form of recurrent NNs \cite{schmidt2019recurrent}, and for graph data with GNNs  \cite{wu2020comprehensive}. GNNs have seen a surge of interest in the classical machine learning community in the past decade \cite{wu2020comprehensive,zhou2020graph}. They are designed to process data that is presented in graph form, like social networks \cite{wu2020comprehensive}, molecules \cite{mansimov2019molecular}, images \cite{long2021graph} or instances of combinatorial optimization problems \cite{cappart2021combinatorial}.

The first attempt to implement a geometric learning model in the quantum realm was made with the quantum convolutional NN in \cite{cong2019quantum}, where the authors introduce a translation invariant architecture motivated by classical convolutional NNs. Approaches to translate the GNN formalism to QNNs were taken in \cite{verdon2019quantum}, where input graphs are represented in terms of a parametrized Hamiltonian, which is then used to prepare the ansatz of a quantum model called a quantum graph neural network (QGNN). While the approach in \cite{verdon2019quantum} yields promising results, this work does not take symmetries of the input graph into account. (However, in an independent work prepared at the time of writing this manuscript, one of the authors of \cite{verdon2019quantum} shows that one of their proposed ansatzes is permutation invariant \cite{larocca2022group}.) The authors of \cite{henry2021quantum} introduce the so-called quantum evolution kernel, where they devise a graph-based kernel for a quantum kernel method for graph classification. Again, their ansatz is based on alternating layers of Hamiltonians, where one Hamiltonian in each layer encodes the problem graph, while a second parametrized Hamiltonian is trained to solve a given problem. A proposal for a quantum graph convolutional NN was made in \cite{zheng2021quantum}, and the authors of \cite{chen2022novel} propose directly encoding the adjacency matrix of a graph into a unitary to build a quantum circuit for graph convolutions. While all of the above works introduce forms of structured QML models, none of them study their properties explicitly from a geometric learning perspective or relate their performance to unstructured ansatzes. 

The authors of \cite{mernyei2021equivariant} take the step to introduce an equivariant model family for graph data and generalize the QGNN picture to so-called equivariant quantum graph circuits (EQGCs). EQGCs are a very broad class of ansatzes that respect the connectivity of a given input graph. The authors of \cite{mernyei2021equivariant} also introduce a subclass of EQGCs called equivariant quantum Hamiltonian graph circuits (EH-QGCs), that includes the QGNNs by \cite{verdon2019quantum} as a special case. EH-QGCs are implemented in terms of a Hamiltonian that is constructed based on the input graph structure, and they are explicitly equivariant under permutation of vertices in the input graph. The framework that the authors of \cite{mernyei2021equivariant} propose can be seen as a generalization of the above proposals. Different from the above proposals, EQGCs use a post-measurement classical layer that performs the functionality of an aggregation function as those found in classical GNNs. In classical GNNs, the aggregation function in each layer is responsible for aggregating node and edge information in an equivariant or invariant manner. Popular aggregation functions are sums or products, as they trivially fulfill the equivariance property. In the case of EQGCs, there is no aggregation in the quantum circuit, and this step is offloaded to a classical layer that takes as input the measurements of the PQC. Additionally, the EQGC family is defined over unweighted graphs and only considers the adjacency matrix of the underlying input graph to determine the connectivity of the qubits. The authors of \cite{mernyei2021equivariant} also show that their EQGC outperforms a standard message passing neural network on a graph classification task, and thereby demonstrate a first separation of quantum and classical models on a graph-based learning task. 

During preparation of our final manuscript, a work on invariant quantum machine learning models was released by the authors of \cite{larocca2022group}. They prove for a number of selected learning tasks whether an invariant quantum machine learning model for specific types of symmetries exists. Their work focuses on group invariance, and leaves proposals for NISQ-friendly equivariant quantum models as an open question.

Our proposal is most closely related to EH-QGCs, but with a number of deviations. First, our model is defined on weighted graphs and can therefore be used for learning tasks that contain node as well as edge features. Second, the initial state of our model is always the uniform superposition, which allows each layer in the ansatz to perform graph feature aggregation via sums and products of node and edge features, as discussed in \Cref{sec:ansatz_structure_equivariance}. Third, we do not require a classical post-processing layer, so our EQC model is purely quantum. Additionally, in its simplest form as used in this work, the number of qubits in our model scales linearly with the number of nodes in the input graph, while the depth of each layer depends on the graph's connectivity, and therefore it provides one answer to the question of a NISQ-friendly equivariant quantum model posed by \cite{larocca2022group}.

\subsection{Neural combinatorial optimization with reinforcement learning}\label{sec:ncorl}

The idea behind NCO is to use a ML model to learn a heuristic for a given optimization problem based on data. When combined with RL, this data manifests in form of states of an environment, while the objective is defined in terms of a reward function, as we will now explain in more detail. In the RL paradigm, the model, referred to as an agent, interacts with a so-called environment. The environment is defined in terms of its state space $\mathcal{S}$ and action space $\mathcal{A}$, that can both either be discrete or continuous. The agent alters the state of the environment by performing an action $a \in \mathcal{A}$, whereafter it receives feedback from the environment in form of the following state $s' \in \mathcal{S}$, and a reward $r$ that depends on the quality of the chosen action, given the initial state $s$. Actions are chosen based on a policy $\pi(a|s)$, which is a probability distribution of actions $a$ given states $s$. The definition of the state and action spaces and the reward function depends on the given environment. In general, the goal of the agent is to learn a policy that maximizes the expected return $G$,

\begin{equation}
    G_t = \sum_{k=0}^{\infty} \gamma^k r_{t+k+1},
\end{equation}
where $\gamma \in [0, 1]$ is the discount factor that determines the importance of future rewards in the agent's decision. The above definition of the expected return is for the so-called infinite horizon, where the interaction with the environment can theoretically go on to infinity. In practice. we usually work in environments with a finite horizon, where the above sum runs only over a predefined number of indices. There are many different approaches to maximize the expected return, and we refer the interested reader to \cite{sutton2018reinforcement} for an in-depth introduction.

In this work, we focus on so-called Q-learning, where the expected return is maximized in terms of Q-values. The values $Q(s, a)$ for each $(s, a)$ pair also represent the expected return following a policy $\pi$, but now conditioned on an initial state $s_t$ and action $a_t$,

\begin{equation}
    Q(s, a) = \mathbb{E}_{\pi}[G_t|s=s_t, a=a_t].
\end{equation}
When the agent is implemented in form of a NN, the goal of the NN is to approximate the optimal Q-function $Q^*$. One popular method to use a NN as a function approximator for Q-learning is called the deep Q-network (DQN), and the resulting DQN algorithm \cite{mnih2015human}. In this algorithm, the NN is trained similarly as in the supervised case, but without a given set of labelled examples. Instead, the agent collects samples at training time by interacting with the environment. These samples are stored in a memory, out of which a batch of random samples is drawn for each parameter update step. Based on the agent's output, the label for a given $(s, a)$ pair from the memory is computed as follows,

\begin{equation}\label{eq:q_update_dqn}
     q = r_{t+1} + \gamma \cdot  \underset{a}{\max}~\hat{Q}_{\boldsymbol{\theta}}(s_{t+1}, a),
\end{equation}
and this label is then used to compute parameter updates.
The update is not computed with the output of the function approximator $Q$, but by a copy $\hat{Q}$ called the target network, which is updated with the current parameters of the function approximator at fixed intervals. The purpose of this target network is to stabilize training, and how often it is updated is a hyperparameter that depends on the environment. In our case, the function approximator and target network are implemented as PQCs, while the parameter optimization is perfomed via the classical DQN algorithm. For more detail on implementing the DQN algorithm with a PQC as the function approximator, we refer the reader to \cite{chen2020variational,lockwood2020reinforcement,skolik2022quantum}.

To evaluate the performance gains of an ansatz that respects certain symmetries relevant to the problem at hand, we apply our model to a practically motivated learning task on graphs. The TSP is a low-level abstraction of a problem commonly encountered in mobility and logistics: given a list of locations, find the shortest route that connects all of these locations without visiting any of them twice. Formally, given a graph $\mathcal{G}(\mathcal{V}, \mathcal{E})$ with vertices $\mathcal{V}$ and weighted edges $\mathcal{E}$, the goal is to find a permutation of the vertices such that the resulting tour length is minimal, where a tour is a cycle that visits each vertex exactly once. A special case of the TSP is the 2D Euclidean TSP, where each node is defined in terms of its $x$ and $y$ coordinates in Euclidean space, and the edge weights are given by the Euclidean distance between these points. In this work, we deal with the symmetric Euclidean TSP on a complete graph, where the edges in the graph are undirected. This reduces the number of possible tours from $n!$ to $\frac{(n-1)!}{2}$. However, even in this reduced case the number of possible tours is already larger than 100k for instances with a modest number of ten cities, and the TSP is a well-known NP-hard problem.

To solve this problem with a RL approach, we follow the strategy introduced in \cite{dai2017learning}. In this work, a classical GNN is used to solve a number of combinatorial optimization problems on graphs. The authors show that this approach can outperform dedicated approximation algorithms defined for the TSP, like the Christofides algorithm, on instances of up to 300 cities. One episode of this learning algorithm for the TSP can be seen in \Cref{fig:qncorl_episodes}, and a detailed description of the learning task as implemented in our work is given in \Cref{sec:qgnn_architecture}.

\subsection{Solving the TSP with the QAOA}
\label{sec:tsp_qaoa}

The QAOA is implemented as a PQC by a Trotterization of Adiabatic Quantum Computation (AQC)~\cite{farhi2014quantum}. In general, for AQC, we consider a starting Hamiltonian $H_{0}$, for which both the formulation and the ground state are well known, and a final Hamiltonian $H_{P}$, that encodes the combinatorial optimization problem to be solved. The system is prepared in the ground state of the Hamiltonian $H_{0}$ and then it is evolved according to the time-dependent Hamiltonian:
$$H(t):=(1-s(t))H_{0}+s(t)H_{P},$$
where $s(t)$ is a real function called annealing schedule that satisfies the boundary conditions: $s(0)=0$ and $s(T)=1$, with $T$ the duration of the evolution. To implement this as a quantum circuit we use the following approximation:
\begin{equation}\label{eq:trotter}
    e^{A+B} \approx \left(e^{\frac{A}{r}} e^{\frac{B}{r}}\right)^{r}, r\rightarrow +\infty,
\end{equation}
which is knwon as the Trotter-Suzuki formula. By using this formula to approximate the evolution according to $H(t)$ and by parameterizing time we obtain:
\begin{equation}\label{eq:gates}
    e^{-i\beta_{p}H_{0}}e^{-i\gamma_{p}H_{P}}\cdots e^{-i\beta_{1}H_{0}}e^{-i\gamma_{1}H_{P}}.
\end{equation}
All of these matrices are unitary since the Hamiltonians in the argument of the exponential are all Hermitian. We define a parameter $p$ (integer known as the depth, or level) of QAOA which has the same role as r in \Cref{eq:trotter}. Increasing the depth $p$ adds additional layers to the QAOA circuit, and thus more closely approximates the H(t)~\cite{farhi2014quantum}.

In QAOA, all qubits are initialized to $|+\rangle^{\otimes_{n}}$, which is the ground state of $H_{0}=\sum_{i}\sigma^{(i)}_{x}$. Alternating layers of $H_p$ and $H_0$ are added to the circuit ($p$ times), parameterized by $\gamma$ and $\beta$ as defined in \Cref{eq:gates}. The values of $\gamma$ and $\beta$ are found by minimizing the expectation value of $H_p$, and thus approximate the optimal solution to the original combinatorial optimization problem. When using QAOA, we do not solve the TSP directly, but a QUBO representation of this problem. This representation is well-known, and can be found in~\cite{lucas2014ising}:

\begin{multline*}
    \sum_{(i,j)\in \mathcal{E}}\sum_{t=1}^{N}\frac{\varepsilon_{i,j}}{W}x_{i,t}x_{j,t+1}+\sum_{i\in \mathcal{V}}\left(1-\sum_{t=1}^{N}x_{i,t}\right)^{2}+\\
    +\sum_{t=1}^{N}\left(1-\sum_{i\in \mathcal{V}}x_{i,t}\right)^{2}+\sum_{(i,j)\notin \mathcal{E}}\sum_{t=1}^{N}x_{i,t}x_{j,t+1}.
\end{multline*}
Here, $\varepsilon_{i,j}$ are the distances between two nodes $i,j \in \mathcal{V}$ and $W:=\max_{(i,j)\in \mathcal{E}}\varepsilon_{i,j}$. The variables $x_{v,t}$ are binary decision variables denoting whether node $v$ is visited at step $t$.
We optimize the $\beta$ and $\gamma$ parameters for $p=1$ by performing a uniform random search over the space $[0, 2\pi]^2$, and selecting the best configuration found.

For $p=2$ and $3$, we optimized the circuit parameters using Constrained Optimization BY Linear Approximation (COBYLA). In addition, similar to \cite{qaoaperf}, we employed a $p$-dependent initialization technique for the circuit parameters. Specifically, $(p+1)$-depth QAOA circuit parameters were initialized with the optimal parameters from the $p$-depth circuit, as follows: 
\begin{align*}
    \boldsymbol \gamma &= (\gamma_{1},\ldots,\gamma_{p'},0),\\
    \boldsymbol \beta &= (\beta_{1},\ldots,\beta_{p'},0).
\end{align*}

This way we are allowing the parameter training procedure to start in a known acceptable state based on the results of the previous step. In \Cref{fig:tsp5_qaoa} we show our results for five-city instances of the TSP. The approximation ratio shown is derived by dividing the tour length of the best feasible solution, measured as the output of the trained QAOA circuit, by the optimal tour length of the respective instance. In addition, we compute results for two different $p=3$ QAOA circuits: the first is trained in the procedure described above (where the parameters are trained for each instance). The second uses the parameters of the best QAOA circuit out of those for all instances evaluated at $p=3$, following a concentration of parameters argument as presented in \cite{brandao2018concentration}. The second method is closer to what is done in a ML context, where one set of parameters is used to evaluate the performance on all validation samples.

Due to the number of qubits required to formulate a QUBO for the TSP, we were not able to run QAOA for all TSP instances. For example, an instance with six cities already requires $25$ qubits (we can fix the choice of the first city to be visited without loss of generality, requiring only $(n-1)^{2}$ variables to formulate the QUBO). A different formulation of the QUBO problem presented in~\cite{glos2020nlogn}, that needs $\mathcal{O}(n\log(n))$ qubits, avoids this issue by modifying the circuit design. However, this proposal increases the circuit depth considerably and is therefore ill-suited for the NISQ era. 

In \Cref{fig:tsp5_qaoa}, we can see that finding a good set of parameters for QAOA to solve TSP is hard even for five-city instances. We note that the performance of QAOA improves with higher $p$, however, QAOA performance is still far from matching the approximation ratios obtained by EQC even for $p=3$, which can be seen in \Cref{fig:tsp5_qaoa} as a black dashed line. Furthermore, we note that significant computational effort is required to obtain these results: methods like COBYLA are based on gradient descent, which requires us to evaluate the circuit many times until either convergence or the maximum number of iterations is reached. We also note that due to the heuristic optimization of the QAOA parameters themselves, we are not guaranteed that the configuration of parameters is optimal, which may result in either insufficient iterations to converge or premature convergence to sub-optimal parameter values. In an attempt to mitigate this, we tested several optimizers (Adam, SPSA, BFGS and COBYLA) and used the best results, which were those found by COBYLA.

\section*{Data availability statement}

The data sets containing the TSP instances studied in this work and their optimal solutions can be found on GitHub (\url{https://github.com/askolik/eqc_for_nco}).

\section*{Code availability statement}

The full code that was used to generate the numerical results and figures in this work can be found on GitHub (\url{https://github.com/askolik/eqc_for_nco}).

\section*{Acknowledgements}
AS thanks Elies Gil-Fuster and Radoica Draskic for valuable discussions about geometric deep learning. AS, MC and SY are funded by the German Ministry for Education and Research (BMB+F) in the project QAI2-Q-KIS under grant 13N15587. VD is supported by the Dutch Research Council (NWO/OCW), as part of the Quantum Software Consortium programme (project number 024.003.037). VD also acknowledges the support by the project NEASQC funded from the European Union’s Horizon 2020 research and innovation programme (grant agreement No 951821). VD also acknowledges the funding by the European Union under Grant Agreement 101080142 and the project EQUALITY.

\section*{Author contributions}

AS conceived the idea for this work. AS and VD conducted theoretical analysis of the proposed ansatz. AS and MC performed numerical simulations and the analysis of their results. SY created the data set used in this work. VD and TB supervised the project. AS, MC and SY wrote the first draft of the manuscript; all authors contributed to editing the final manuscript.

\section*{Competing interests}

The authors declare no competing interests.

\bibliography{lib}
\bibliographystyle{naturemag}

\onecolumngrid
\appendix

\section{Supplementary discussion}

\subsection{Proof of equivariance of ansatz}
\label{sec:proof_equivariance}

\begin{proof}[Proof of Theorem 1]
We first recall the definitions from Theorem 1. Let the ansatz of depth $p$ be of the type as defined in Equation 1 with initial state $\ket{+}^{\otimes n}$ and parameters $\boldsymbol \beta, \boldsymbol \gamma \in \mathbb{R}^p$, that represents an instance of a graph $\mathcal{G}$ with nodes $\mathcal{V}$ and the list of edges $\mathcal{E}$ with corresponding edge weights $\varepsilon_{ij}$, and node features $\boldsymbol \alpha \in \mathbb{R}^n$ with $n=|\mathcal{V}|$. Let $\sigma$ be a permutation of the vertices in $\mathcal{V}$, $P_{\sigma} \in \mathbb{B}^{n \times n}$ the corresponding permutation matrix that acts on the weighted adjacency matrix $A$ of $\mathcal{G}$, and $\Tilde{P}_{\sigma} \in \mathbb{B}^{2^n \times 2^n}$ a matrix that maps the tensor product $\ket{v_1} \otimes \ket{v_2} \otimes \dots \otimes \ket{v_n}$ with $\ket{v_i} \in \mathbb{C}^2$ to $\ket{v_{\Tilde{p}_{\sigma}(1)}} \otimes \ket{v_{\Tilde{p}_{\sigma}(2)}} \otimes \dots \otimes \ket{v_{\Tilde{p}_{\sigma}(n)}}$.

We want to prove that our ansatz is equivariant under permutations of the nodes of the input graph $\mathcal{G}(\mathcal{V}, \mathcal{E})$,
 \begin{equation}\label{eq:restate_equivariance}
     |\mathcal{E}_{(P_{\sigma}^T A P_{\sigma})}, P_{\sigma}^T \boldsymbol \alpha, \boldsymbol \beta, \boldsymbol \gamma \rangle_p = \Tilde{P}_{\sigma} |\mathcal{E}_A,  \boldsymbol \alpha, \boldsymbol \beta, \boldsymbol \gamma \rangle_p.
 \end{equation}
 For this, we have to prove that the unitaries that are used to construct the full circuit are permutation equivariant, i.e.,
 \begin{equation}
     \Tilde{P}_{\sigma} U_{G}(\mathcal{E}_A, \gamma_l) \Tilde{P}_{\sigma}^{\dagger} = U_{G}(\mathcal{E}_{(P_{\sigma}^T A P_{\sigma})}, \gamma_l) 
 \end{equation}
 and
 \begin{equation}
     \Tilde{P}_{\sigma} U_{N}(\boldsymbol \alpha, \beta_l) \Tilde{P}_{\sigma}^{\dagger} = U_{N}(P_{\sigma}^T \boldsymbol \alpha, \beta_l).
 \end{equation}
 We begin with the edge-encoding unitary $U_G$:
 \begin{align}
     &\Tilde{P}_{\sigma} U_{G}(\mathcal{E}_A, \gamma_l) \Tilde{P}_{\sigma}^{\dagger} = \Tilde{P}_{\sigma} e^{-i \gamma_l H_{\mathcal{G}}} \Tilde{P}_{\sigma}^{\dagger}\\
     &= e^{-i \gamma_l \Tilde{P}_{\sigma} H_{\mathcal{G}}\Tilde{P}_{\sigma}^{\dagger}}\\
     &= e^{-i \gamma_l H_{\mathcal{G}_{(P_{\sigma}^T A P_{\sigma})}}}\\
     &= U_{G}(\mathcal{E}_{(P_{\sigma}^T A P_{\sigma})}, \gamma_l),
 \end{align}
 where line (5) holds because for any unitary $U$ we have $ U e^{-iH_{\mathcal{G}}} U^\dagger = e^{-i U H_{\mathcal{G}} U^\dagger}$, and line (6) holds because $H_{\mathcal{G}} = \sum_{\varepsilon \in \mathcal{E}} \varepsilon_{ij} Z_i Z_j$ is defined completely through the adjacency matrix and the edge weights of the input graph $\mathcal{G}$, and $\Tilde{P}_{\sigma}$ and $P_{\sigma}$ are defined through permutations $\sigma$ on the nodes of $\mathcal{G}$. Similarly, we get
 \begin{align}
     &\Tilde{P}_{\sigma} U_{N}(\boldsymbol \alpha, \beta_l) \Tilde{P}_{\sigma}^{\dagger} = \Tilde{P}_{\sigma}  \bigotimes_i^{|\mathcal{V}|} \mathrm{Rx}(\alpha_i, \beta_l) \Tilde{P}_{\sigma}^{\dagger}\\
     &= \Tilde{P}_{\sigma}  \bigotimes_i^{|\mathcal{V}|} \mathrm{exp}\left(-i \frac{\alpha_i \beta_l}{2} X \right) \Tilde{P}_{\sigma}^{\dagger}\\
     &= \bigotimes_i^{|\mathcal{V}|} \mathrm{exp}\left(-i \frac{\alpha_{\sigma^{-1}(i)} \beta_l}{2} X \right) \\
     &= U_{N}(P_{\sigma}^T \boldsymbol \alpha, \beta_l).
 \end{align}
 As each of the unitaries in the circuit is equivariant under permutations of the graph nodes, and the initial state is trivially permutation invariant $\ket{+} = \Tilde{P_{\sigma}} \ket{+}$, we arrive at \Cref{eq:restate_equivariance}.

\end{proof}

In the main text, we have briefly discussed scenarios in which each node in the graph corresponds to multiple qubits in the ansatz. The above proof holds for the one qubit case, and whether an ansatz with multiple qubits per node will still be equivariant depends on the chosen mapping of node and edge features.

\subsection{Proof that ansatz can generate optimal tours for rationally independent set of edge weights}
\label{sec:opt_containing_ratin}

In this section, we show how the fact that the sine function can approximate an arbitrary set of rationally independent values $\{x_1, \dots, x_n\}$ with labels $y_i \in [-1, 1]$ can be used prove that our ansatz at depth one can construct the optimal tour for a graph with edge weights that are rationally independent.

\begin{theorem}[Ansatz can generate optimal tours for rationally independent edge weights]
    There exists a setting $(\beta, \gamma)^*$ for each graph instance of the symmetric TSP such that the ansatz at depth one described in Equation 1 will produce the optimal tour $T^*$ with the node selection process described in Definition 2, given that the edge weights $\varepsilon_{ij}$ of the graph are rationally independent and $\varepsilon_{ij} \gamma \neq \frac{\pi}{4} + n \pi ~\forall~ n \in \mathbb{Z}$.
\end{theorem}

\begin{proof}
As known from \cite{harman2015sin}, we can find a parameter $\omega$ such that we can approximate an arbitrary labeling in $[-1, 1]$ for our rationally independent edge weights with the sine function. Given that this labeling exists, we now show how to use this labeling to generate the optimal tour with the EQC at depth one.

For $p=1$, we can compute the analytic form of the expectation values of our circuit as defined in Equation 12 and Equation 13 as the following, by a similar derivation as in \cite{ozaeta2022expectation},
\begin{equation}\label{eq:anal_expectation_pr}
    \langle O_{v_l} \rangle = \varepsilon_{v_{t-1}, v_l} \cdot \mathrm{sin}(\beta \pi) ~\mathrm{sin}(\varepsilon_{v_{t-1}, v_l} \gamma)
    \cdot \prod_{\underset{k \neq v_{t-1}}{(v_l,k) \in \mathcal{E}}} \mathrm{cos}(\varepsilon_{v_l, k} \gamma),
\end{equation}
where $v_{t-1}$ is the last node in the partial tour and $v_l$ is the candidate node. By the identity $\mathrm{cos}(\theta) = \mathrm{sin}(\frac{\pi}{2} - \theta)$ we can rewrite \Cref{eq:anal_expectation_pr} as
\begin{equation}\label{eq:anal_expectation_sins}
    \langle O_{v_l} \rangle = \varepsilon_{v_{t-1}, v_l} \cdot \mathrm{sin}(\beta \pi) ~\mathrm{sin}(\varepsilon_{v_{t-1}, v_l} \gamma)
    \cdot \prod_{\underset{k \neq v_{t-1}}{(v_l,k) \in \mathcal{E}}} \mathrm{sin}\left(\frac{\pi}{2} - \varepsilon_{v_l, k} \gamma \right).
\end{equation}
Let us now assume that we want to construct a fixed (but arbitrary) tour $T$. First, we notice that the term $\mathrm{sin}(\beta \pi)$ does not depend on $v_{t-1}$ or $v_l$ and is the same for all $v_l$. This means that this term can merely flip the sign of all $\langle O_{v_l} \rangle$, and from now on w.l.o.g.\ we assume that $\beta$ is such that the term is positive. Now we can again formulate the tour generation task in terms of a binary classification problem, where we want to find a configuration of labels for our remaining sin terms in \Cref{eq:anal_expectation_sins} s.t.\ the product will have the highest expectation value in each node selection step for the edge that produces the ordering we have chosen for $T$. Again, we can accomplish this for arbitrary settings of edge weights by only considering the sign of the resulting product. This means that we have to find an assignment of the edges $\varepsilon_{ij}$ to the classes $f_{\pm}$ that at each step of the node selection process will lead to the node being picked that we specify in $T$. As all edges can occur in the above products multiple times during the node selection process, this is a non-trivial task. However, if we can  guarantee that each $\langle O_{v_l} \rangle_t$ at node selection step $t$ contains at least one \textit{unique term} that is only present in this specific expectation value, we can use this term to control the sign of this specific value. Each $\varepsilon_{ij}$ occurs either in the leading term $\mathrm{sin}(\varepsilon_{ij} \gamma)$ (corresponding to the candidate edge to be potentially added in the next step) or in the product term as $\mathrm{sin}(\frac{\pi}{2} -  \varepsilon_{ij} \gamma)$ (corresponding to an outgoing edge from the current candidate). We can easily see that the leading term \textit{only} appears in the case when we ask for this specific $\varepsilon_{ij}$ to be the next edge in the tour, and from Definition 2 we know that this only happens once in the node selection process. In all other expectations, $\varepsilon_{ij}$ appears only with the ``offset" of $\frac{\pi}{2}$. This means that this leading term is the unique term that we are looking for, as long as $\mathrm{sin}(\varepsilon_{ij} \gamma) \neq \mathrm{sin}(\frac{\pi}{2} -  \varepsilon_{ij} \gamma)$, so as long as $\mathrm{sin}(\varepsilon_{ij} \gamma) \neq \mathrm{cos}(\varepsilon_{ij} \gamma)$. We know that $\mathrm{cos}(\theta) = \mathrm{sin}(\theta)$ for $\theta = \frac{\pi}{4} + n \pi$ with $n \in \mathbb{Z}$. So as long as 
\begin{equation}\label{eq:param_condition_universal}
    \varepsilon_{ij} \gamma \neq \frac{\pi}{4} + n \pi ~\forall~ n \in \mathbb{Z}, \varepsilon_{ij} \in \mathcal{E},
\end{equation}
and all $\varepsilon_{ij}$ are unique, our ansatz can construct the desired tour $T$. In this case, we have a guarantee that we can construct the tour $T$ for any configuration of edges that fulfills \Cref{eq:param_condition_universal}. In particular, this means that we can construct the optimal tour in this way.
\end{proof}

Additionally, due to equivariance of our model under node permutations, it is sufficient to find just one of these settings for any tour $T$ and use this to construct the optimal tour. We can do this because once we have found a labeling of the edges for the classes $f_{\pm}$ that produces a given tour sequence, we can simply reshuffle our nodes s.t.\ this sequence of the ``physical nodes" (indices corresponding to qubits) will produce the optimal tour for our ``logical nodes" ((x, y) coordinates of graph nodes). However, note that this can not be used to trivially find the optimal tour, as it amounts to trying all $\frac{(n-1)!}{2}$ possible tours.

\section{Supplementary Figures}
\label{sec:means_p_vals}

To make statements on the statistical significance of the difference between the performance of the EQC and NEQC shown in Figure 5 of the main text, we perform a two-sample t-test on the two models for the same instance sizes (i.e., for the data in the two boxes for each instance size) with the null hypothesis that the averages of the two distributions are the same. Based on this, we compute p-values to quantify the statistical significance of the differences between models. 

\begin{figure}
\includegraphics[width=0.5\textwidth]{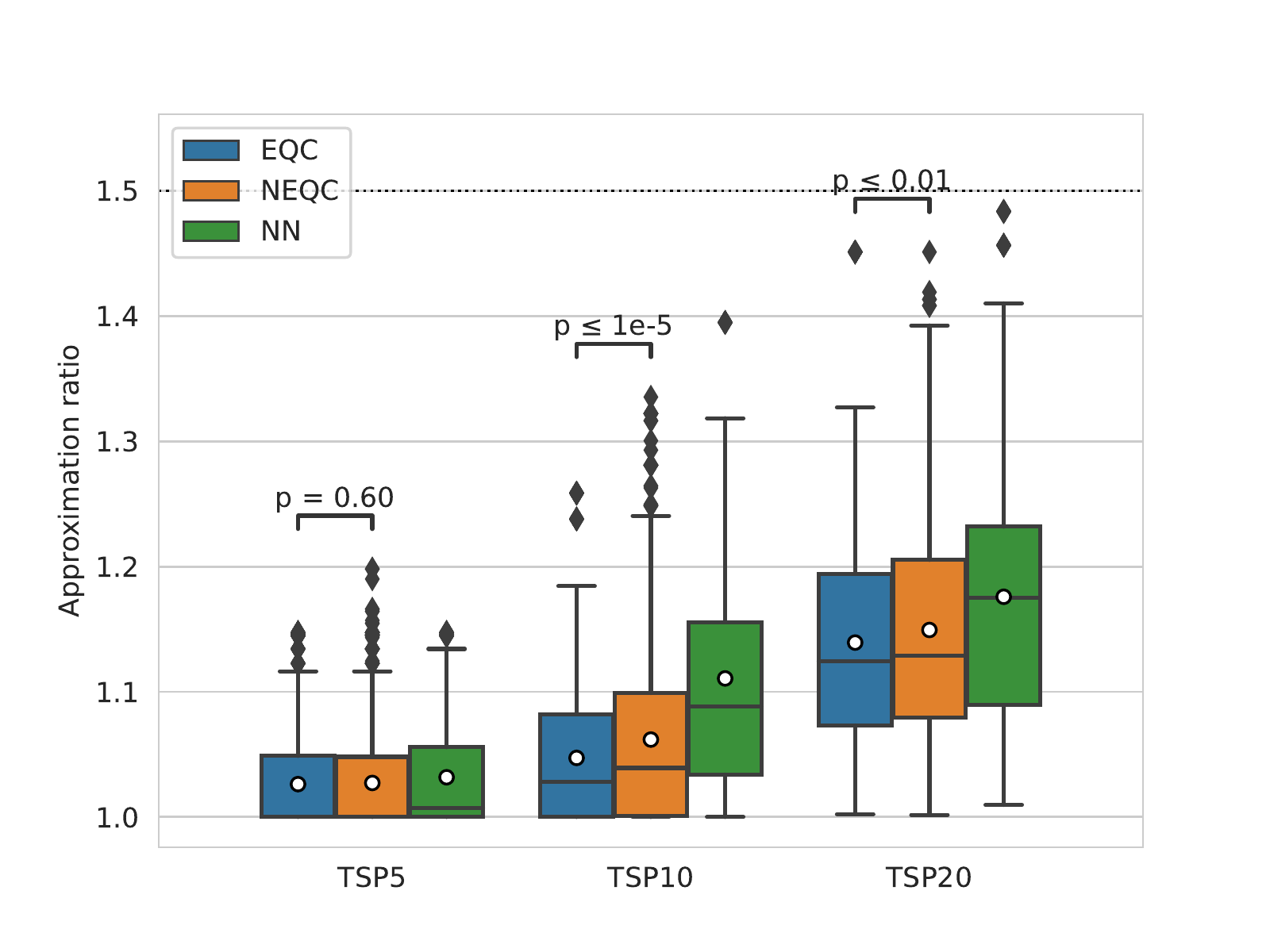}
\caption{P-values for comparison of EQCs and NEQCs at depth one from Figure 5 b).}
\label{fig:pvals_1layer}
\end{figure}

Supplementary \Cref{fig:pvals_1layer} shows p-values for the depth-one EQCs and NEQCs from Figure 5 b). For the 5-city instances, we can not reject the null hypothesis. Indeed, it is already visible by looking at the boxes that the distributions are very similar, which can be expected as the number of permutations of a graph with five vertices is small. However, as we scale up the instance size to ten cities, the corresponding p-value is much smaller than 0.05, which means that we can reject the null hypothesis that the two distributions have the same average with high confidence. This is also the case for the instances with twenty cities, where the p-value is less than 0.01. 

Supplementary \Cref{fig:pvals_4layers} shows p-values for the depth four EQCs and NEQCs from Figure 5 d). Again, the p-value of the 5-city instances is very high with 0.74, so that we can not reject the null hypothesis. Also similarly to the above, the p-values get smaller as we scale up the instance size. For the depth-four ansatzes, the p-value is smallest for the twenty city instances, with a value much smaller than 0.05.

\begin{figure}
\includegraphics[width=0.5\textwidth]{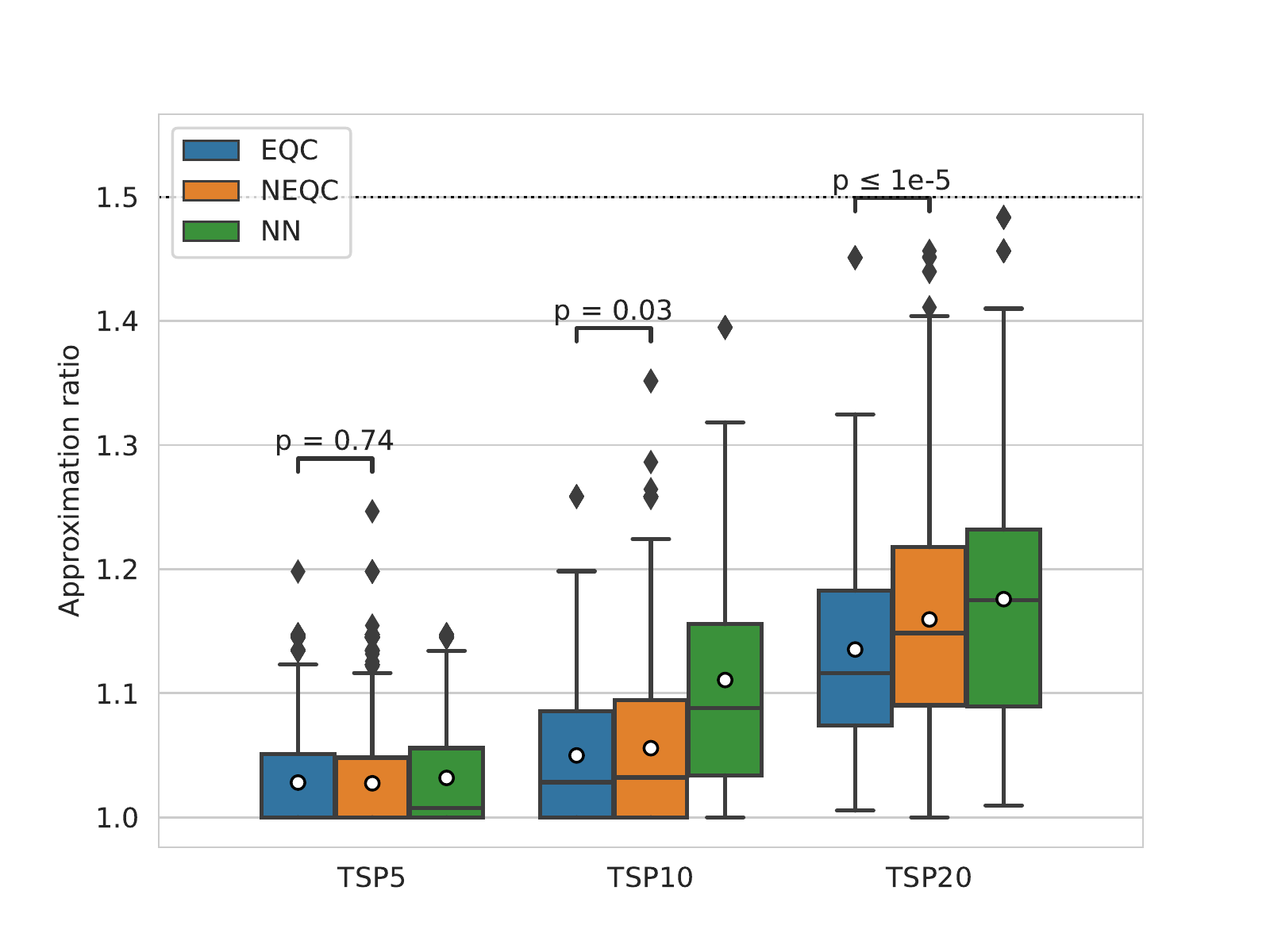}
\caption{P-values for comparison of EQCs and NEQCs at depth four from Figure 5 d).}
\label{fig:pvals_4layers}
\end{figure}

\begin{figure}
\includegraphics[width=0.5\textwidth]{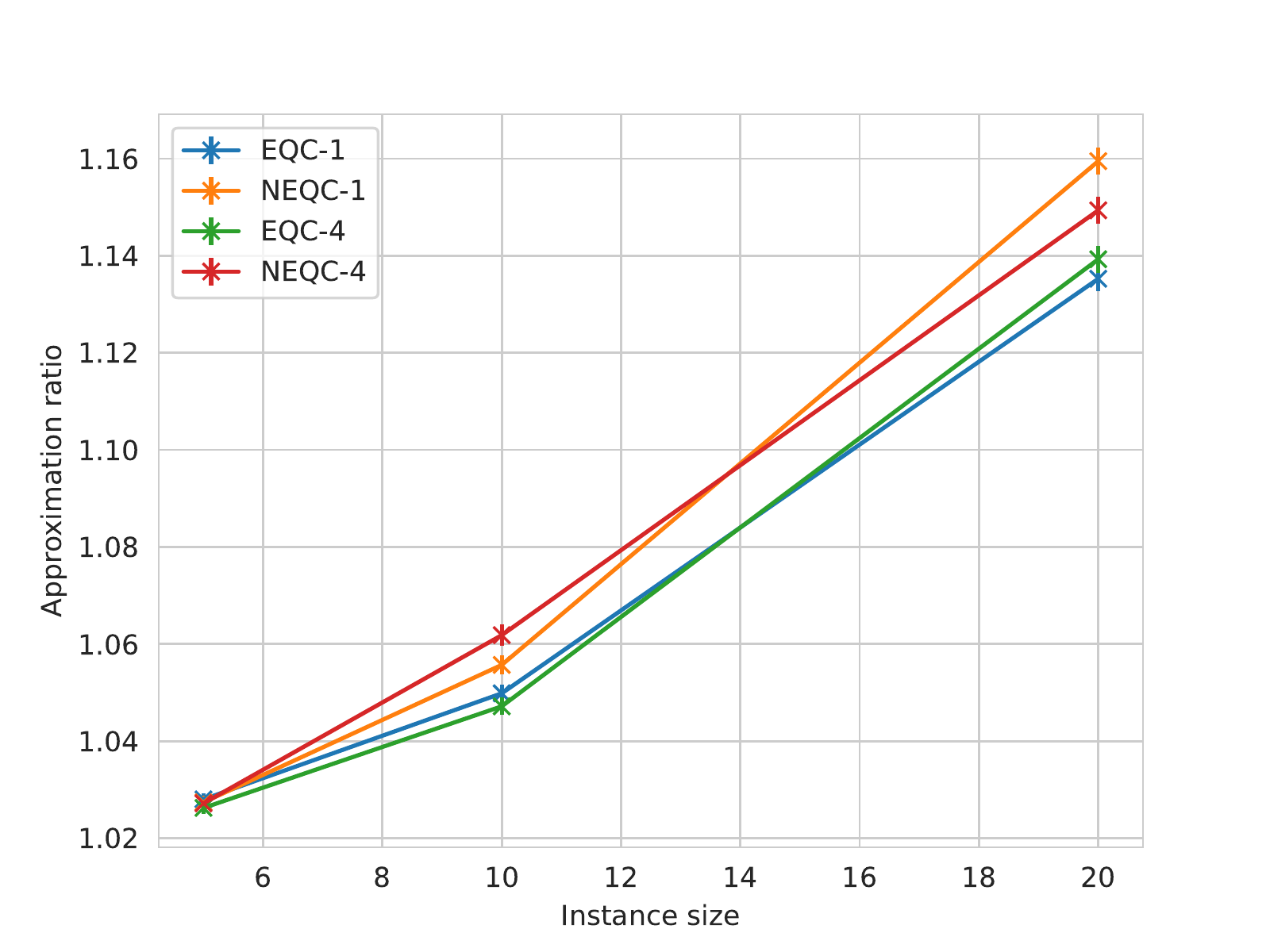}
\caption{Mean and standard error of the mean for the one- and four-layer EQCs and NEQCs in Figure 5 b), d).}
\label{fig:means_eqc_neqc}
\end{figure}
To provide additional insight, we also plot the means and their standard error for both the 1-layer (EQC-1, NEQC-1) and 4-layer (EQC-4, NEQC-4) models in Figure 5. As a rule of thumb, one can expect that when the error bars given by the standard errors of two means do not overlap, the p-value can be smaller than 0.05, while in the case that they do overlap, the p-value is likely much larger. The error bars in Supplementary \Cref{fig:means_eqc_neqc} are in line with this statement, where we see that the error bars for the five-city instances overlap for both circuit depths, while this is not the case for the larger instance sizes and in addition the distance between the means increases for those instance sizes. Remarkably, we also see that the difference between the EQC at depths one and four is very small, and that increasing the circuit depth does not provide much benefit on this learning task.
\twocolumngrid

\end{document}